\documentclass[]{article}
\usepackage{adjustbox}
\usepackage{booktabs}
 \usepackage{epsfig}
\usepackage{caption}
\usepackage{threeparttable}
\usepackage{subcaption}
\usepackage{epstopdf} 
\usepackage{epsfig}
\usepackage{import}
\usepackage{natbib}
\usepackage{amssymb}
\usepackage{graphicx}
\usepackage{xurl}
\usepackage{float}
\usepackage{blindtext}

\newtheorem{proposition}{Proposition}

\usepackage{lineno,hyperref}
\newcommand{\e}[1]{{\mathbb E}\left[ #1 \right]}
\usepackage[title,toc,titletoc,page]{appendix}

\newenvironment{proof}[1][Proof]{\noindent\textbf{#1.} }{\ \rule{0.5em}{0.5em}}
\usepackage{changepage}
\usepackage{multirow}
\usepackage{authblk}
\usepackage{amsfonts}
\usepackage{dsfont}
\usepackage{amsmath,amssymb}
\usepackage{amssymb}
\usepackage{graphicx}
\usepackage{comment}
\usepackage{bbold}

\usepackage[a4paper,total={6in,10in}]{geometry}

\author{Giulia Livieri$^{a,}$\footnote{Corresponding author: G. Livieri, E-mail: g.livieri@lse.ac.uk}  \qquad Davide Radi$^b$ \qquad Elia Smaniotto$^c$ }
\font\myfont=cmr20 
\title{\vspace*{0.8in}\myfont{Pricing Transition Risk with a Jump-Diffusion Credit Risk Model: Evidences from the CDS market}\vspace*{0.2in}}

\begin{document}
\maketitle
\thispagestyle{empty} \medskip
\date{ \medskip \medskip $^a$\textit{The London School of Economics and Political Science, London, United Kingdom.}\medskip \\    \indent    
$^b$\textit{Department of Mathematics for Economic, Financial and Actuarial Sciences, Catholic} \\  \indent \indent \medskip 
\textit{University of Sacred Heart, Milan, Italy.} \\    \indent
$^c$\textit{Department of Economics and Management, University of Florence, Florence, Italy.}}
\bigskip
\indent  \medskip

\bigskip
\bigskip
\bigskip
\begin{abstract}
Transition risk can be defined as the business-risk related to the enactment of green policies,
aimed at driving the society towards a sustainable and low-carbon economy.
In particular, the value of certain firms’ assets can be lower because they need to transition to a less carbon-intensive business model.
In this paper we derive formulas for the pricing of defaultable coupon bonds and Credit Default Swaps to empirically demonstrate that a jump-diffusion credit risk model in which the downward jumps in the firm value are due to tighter green laws can capture, at least partially, the transition risk. 
The empirical investigation consists in the model calibration on the CDS term-structure, performing a quantile regression to assess the relationship
between implied prices and a proxy of the transition risk.
Additionally, we show that a model without jumps lacks this property, confirming the jump-like nature of the transition risk.
\end{abstract}

\newpage
\section{Introduction}
The European Union (\textrm{EU}) aims to be climate-neutral by 2050 -- an economy with net-zero greenhouse gas emissions. This objective is at the heart of the \textrm{European Green Deal} and in line with the EU's commitment to global climate action under the 2015 \textrm{Paris Climate Agreement}. However, while the transition to low-carbon economies is a policy imperative, the path may entail extensive policy, legal, technology, and market changes. In particular, the path could expose financial markets to the so-called \textit{transition risk}\footnote{Notice that another type of risk cited in climate finance is the so-called \textit{physical-risk}, which reflects the uncertain economic costs and financial losses from tangible climate-related
adverse trends and more severe extreme events}.

Transition risk results from the adjustment towards a low-carbon economy and the possibility that shifts in policies or technologies designed to mitigate and adapt to climate change could affect the value of financial assets and liabilities, disrupting inter-mediation and financial stability. More precisely, transition risk results from changes in climate policy or regulation or technological innovations that cause a decrease in the competitiveness of high-carbon technologies and infrastructures (in turn leading to increased costs, stranded assets, stranded processes, or credit losses). The EU has been gradually reducing its level of \textit{greenhouse gases} (GHG, henceforth) by over 20$\%$ in 2020 since its 1990 levels, mainly thanks to the creation in 2005 of the EU Emissions Trading Scheme (ETS). However, concerns over climate-related financial risks only recently became more prevalent in the financial sector. Therefore, a relevant \textit{research question} regards the ability to account for transition risk by \textit{credit-risk models}.

Traditionally not crafted for capturing transition risk, credit-risk models account for different credit risks drivers, such as equity returns, equity volatility, and macroeconomic factors (often proxied by stochastic risk-free interest rates). More sophisticated credit-risk models are based on \cite{merton1974pricing}-like \textit{jump-diffusion processes} to capture unexpected events that negatively affect equity value. In particular, the jump component is required to replicate the empirically-observed high short-term credit spreads (see, e.g., \cite{Zhou2001}, among others). In addition, it is economically justified by the discounting of unforeseen information of firms' credit quality deriving by the asymmetric information that affects investors; see, e.g., \cite{Giesecke2006}. Since asymmetric information also characterizes the introduction, and the time of introduction, of green regulations, a question arises if jump-diffusion models are suitable to capture transition risk.

Recently, \cite{AgliardiAgliardi2021} propose a model for defaultable \textit{perpetual bonds} incorporating uncertainty due to climate-related risks, in addition to the uncertainty about corporate earnings; the \cite{AgliardiAgliardi2021}'s model is similar to the \cite{hilberink2002optimal} one. In particular, in their specification, the random policy shocks and \textit{degree of greenness} determine downward jumps in the firm value. Compared to these references, our contribution is to consider a version of their credit-risk model with \textit{finite-time maturity} and an \textit{exogenous} default barrier instead of endogenous, and derive a closed-formula for pricing \textit{Credit Default Swap} (CDS, henceforth). Here, we also mention the work of \cite{le2022corporate}, who quantify in the previous models the impact of scenario uncertainty and progressive information discovery on bond prices. 

In particular, we focus on CDS spreads because they offer some advantages with respect to usually used credit risk measures (e.g., ratings). Indeed: (i) They are more reactive to new information arriving in the market than bonds or ratings (see, e.g., \cite{norden2009co}); (ii) They admit standardized contractual characteristics which make them comparable within the corporate universe; (iii) Contrary to bonds, the trading in the CDS market is sufficiently active (see, e.g., \cite{ederington2015bond}); see also the discussion in  \cite{BlasbergKieselTaschini2022}. In particular, CDS spreads are a more precise measure of default risk premia than interest rates spreads, which also include other risk factors such as inflation expectations, exchange rate differences, and heterogeneity in the risk free-rate across countries.

We test the ability of the proposed model to capture transition risk. More precisely, we inspire our empirical investigation from \cite{transitionriskCDSTaschini2021} and \cite{BlasbergKieselTaschini2022}; a more detailed description of their methodology will be given in Section \ref{Sec::4}. They propose a forward-looking proxy for transition risk exposure based on CDS spreads and study how it affects firms' creditworthiness. 
A second driver of transition risk is constructed adopting carbon allowances, tradable government permits that allow firms to emit $\mbox{CO}_2$ in the atmosphere. As quoted contracts, the price of carbon allowances (i.e.,~\textit{carbon price}) should increase when new green policies are enacted, affecting the financial stability of high-pollutant companies. Hence, we aim to link transition risk to the carbon price time series. 

In our investigation, we consider a sample of 84 European companies. Then, we calibrate the model by fitting the market term structure of CDS spreads for every single company and each single trading day in the 5 years time window. Given the very long time horizon of the low-carbon transition, we consider market-implied credit risk for a variety of maturities, namely the 6-month, 1-year, 2-year, 3-year, 4-year, 5-year, 7-year, 10-year, 20-year, and 30-year CDS spread. At this point, we conduct a quantile regression analysis to investigate to which extent the transition risk can explain the spreads between daily variation in theoretical CDS quotations and market-model calibration errors. The quantile regression allows us to provide a more detailed picture of the observed effects with respect to conditional mean models. In addition to the transition risk factor, we include in the regression the most relevant determinants of CDS spreads analyzed in the literature, from firm-specific variables like stock return and volatility, to common variables like the underlying market condition. The investigation, considering as proxy for the transition risk the forward-looking metrics proposed in \cite{transitionriskCDSTaschini2021} and \cite{BlasbergKieselTaschini2022}, confirms the hypothesis that the jump-diffusion model can depict the CDS market stylized facts and, moreover, no significant relationship is measured between calibration errors and transition risk.
Considering the second driver of the transition risk (i.e. \textit{carbon price}), we found discordant results, since it does not follow theoretical expectations about the relationship between transition risk and CDS quotations, thus it is omitted from the final results.
To conclude, we conduct the same analysis using the model without jumps, showing that it is not capable of depicting the stylized facts of the CDS market, confirming the  jump-like nature of the transition risk.
\subsection{Positioning of our Work in the Literature}\label{subsec:literature}
The present work contributes to the literature aiming to establish the capability of a credit-risk jump-diffusion model to capture transition risk from the market. In particular, it stands among works focusing on the investigation of topics such as sustainable finance, climate change, and transition risk.
In addition to the references mentioned so far, we cite the following ones. \cite{delis2018being} observe that banks appeared to start pricing climate policy risk after the Paris Climate Agreement. \cite{jung2018carbon} provide evidences of the existence of a positive association between the cost of debt and carbon related risks for firms. \cite{battiston2017climate} find that while direct exposures to the fossil fuel sector are small, the combined exposures to climate
policy-relevant sectors are large, heterogeneous, and amplified by large indirect exposures via financial counterparties. \cite{ilhan2021carbon} show that, for a sample of S$\&$P 500 companies, higher emissions increase downside risk -- the potential losses that may occur if a particular investment position is taken. \cite{monasterolo2020blind} indicate that investors require higher risk premia for carbon-intensive industries’ equity.

 Within the CDS framework, instead, \cite{barth2022esg} enhance well-established fundamental models with ESG data and shows that the market valuation of environmental performance predominantly drives changes in
CDS spreads. \cite{kolbel2020does} construct textual, forward-looking measures of climate risk exposure and are able to show that transition risk is priced in CDS market. \cite{BlasbergKieselTaschini2022} propose and implement a novel market-based measure of exposure to transition risk and show that the transition risk factor is a relevant determinant of CDS spreads. \cite{vozian2022climate} finds
that firms with higher GHG emissions have higher CDS-implied credit risk, even at the 30-year horizon, employing several proxies of the transition risk.

\subsection{Organization of the Paper}
The rest of the paper is organized as follows. Section \ref{Sec::3} introduces the model and discusses the interpretation of parameters as transition risk drivers. 
Section \ref{Sec::4} lays down the semi-closed formulas to price defaultable bonds and CDSs. Section \ref{Sec::DatasetAndResearchDesign} presents the research question(s), the dataset, the control variables and the methodology. 
Section \ref{sec:empirical} states the empirical investigation and discusses the results. Section \ref{Sec::6} reports the conclusions. \\
The derivation of the pricing formulas is provided in Appendix \ref{AppendixA}. The CDS pricing formula for the no-jump credit risk model is mentioned in Appendix \ref{AppendixB}. The sensitivity analysis of the theoretical credit spread generated by the model is discussed in Appendix \ref{Appendix:theoreticalspread}.
In Appendix \ref{Appendix:Results}, a comprehensive representation of the results of the investigation is proposed. Finally, the descriptive statistics of the dataset is reported in Appendix \ref{Appendix:statistics}.

\section{The Model}\label{Sec::3}
Let us consider a structural credit risk model, where the value (EBIT) of a firm is described by the jump-diffusion stochastic process defined on filtered probability space 
$\left(\Omega,\mathcal{F},\left(\mathcal{F}_{\tau}\right)_{t\leq \tau\leq T},Q\right)$. The dynamic of the firm is given by:\footnote{The notation $V_{t}^{-}$ indicates that, if a jump occurs, the value of the process before the jump is used as a left-side contribution.}
\begin{equation} \label{eq:firmAgliard}
    \frac{dV_t}{V^{-}_{t}} = (r - \lambda \xi) dt + \sigma dW^{Q}_t + d \left( \sum_{i=1}^{N_t} ( S_i -1 ) \right)
\end{equation}
where $r$ is the constant free-risk interest rate, $\sigma$ is the volatility and $W_t$ is a Wiener process defined on the jump-related risk-adjusted measure $Q$. 
The Poisson process $N_t$ has a constant arrival rate $\lambda$ and $S_i$ indicates a sequence of non-negative independent and identically distributed (i.i.d.) random variables. 
The amplitude of the jumps, driven by the random variable $Y_i = \ln(S_i)$, follows an exponential distribution, having density:
\begin{equation}
 f_Y (y) = \eta e^{\eta y}   \mathds{1}_{ \{ y < 0 \} }
\end{equation}
where the parameter $\eta$ models the average jumps magnitude and $\mathds{1}_{ \{ . \}}$ is the indicator function, imposing that the firm can only be affected by downward jumps. As a common assumption, the drift of firm-value stochastic process defined in \eqref{eq:firmAgliard} is adjusted by $\lambda \xi$, where:
\begin{equation}
\xi = \e{e^{Y_{i}} -1}=\frac{\eta}{\eta + 1} - 1 
\end{equation}
 so that the dynamics of $V$ is consistent with the jump-related risk-adjusted measure $Q$.
 
In this framework, based on the \cite{AgliardiAgliardi2021}, the jump component in the process describes the impact of green policies on the value of the firm. As new green laws are released by the regulator, imposing lower carbon emissions and promoting the adoption of sustainable technologies, companies are forced to comply with the new requirements,  incurring sudden costs which can undermine the financial stability of the firm and badly affect the creditworthiness.
The impact of the green policies on the value of the firm depends on the parameter $1/\eta$, describing the average jump magnitude. 
This parameter provides a measure of the costs a company must undertake when new green laws are enacted, indicating the exposure to transition risk. 

Firms' exposure to transition risk may result from several factors, such as a high dependency on carbon emissions or to be reliance on brown business models. For instance, companies that are investing in green projects, and aligning their business line to low-carbon standards, will be less exposed to this risk, thus the impact of new policies on their valuation will be minor. In our model, the parameter $\eta$ is associated with a firm's level of greenness, so that greener firms $( \uparrow \eta)$ will be affected by smaller jumps than brown firms $( \downarrow \eta)$.

The intensity of unexpected events that cause a sudden reduction of the firm's value is defined by $\lambda$, the arrival rate of the Poisson process. It reflects the intensity of the transition risk due to climate-change policies. Standard choices comprise considering these values as a constant and assigning it to an exogenous variable such as the rate of enactment of new green laws or an indicator of the rate of climate change or a macroeconomic variable. In the current framework, the arrival rate is defined constant as, for example, in \cite{AgliardiAgliardi2021}. 

When a company is unable to comply with green laws by adjusting the business model, the firm may go bankrupt. In our model, we define this credit event (default) as the first time the firm-value process breaches (from above) a constant default barrier, denoted by $V_{def}$. Therefore, the default event coincides with the following first-passage time:
\begin{equation}
t_{d} = \inf \left\{ t > 0 : V_{t} \leq V_{def} \right\}, \quad t \in (0,T]
\end{equation}
In the case of bankruptcy, the debt holder is refunded with a percentage (so-called recovery rate) of the face value of the security. The recovery rate is assumed constant and denoted by $\Upsilon$. This type of recovery is known as the recovery of face value, see \cite{DuffieSingleton1999}, \cite{GuoJarrowLin2008} and \cite{GunduzUhrigHomburg2014}.

\section{Pricing Defaultable Coupon Bonds and Credit Default Swaps}\label{Sec::4}
In this section, we present the pricing formulas for a defaultable coupon bond and a CDS derived from the model, described in Section \ref{Sec::3}. The approach, based on \cite{ChenKou2009}, has been already proposed in the literature, in fact our model is a simplified version of the \cite{Koumodel} one. Originally, the model was driven by the same dynamics but equipped with two different jumps, providing an upward and downward component. 
For our purpose, as mentioned in the previous section, we intend to have the downward component only. 
Hence, in this framework, it is possible to derive the first-time passage probability for the model, as described in \cite{KouWang2003}. Once obtained, the pricing formulas can be derived accordingly to \cite{ChenKou2009}, requiring the computations of inverse Laplace transforms. At the end of the section, we also propose an overview of the numerical method to derive the inverse of the Laplace transforms.
\subsection{Preliminary Results}
Before introducing the pricing formulas for a defaultable coupon bond and a CDS, let us state some preliminary results. Employing the notation  $x = \ln(V)$, the default probability in $( t,T ]$ is given by: 
\begin{equation}\label{defaultprob}
\displaystyle D\left(x,t;T\right) =  \mathbb{E}\left[\left. \mathds{1}_{\left\{T\geq t_{d} \geq t\right\}} \right|\mathcal{F}_{t}\right] 
\end{equation}  
Following the approach of \cite{KouWang2003}, which is consistent within our model\footnote{In \cite{KouWang2003}, the Laplace transform of the first time passage model has been derived for a jump-diffusion model with a downward and upward component. For our model, which is included in the mentioned version as it has only one jump component, the Laplace transform (reported in \eqref{eq:laplacefirsttimepass}) can be derived in a similar manner. For the sake of brevity, we do not report the proof.}, we can define the Laplace transform with parameter $\omega > 0$ of the default probability \eqref{defaultprob}, that is $H\left(x,t;\omega\right)=\mathcal{L}\left(D\left(x,t;T\right)\right)\left(\omega\right)$, as:
\begin{equation}\label{eq:laplacefirsttimepass}
\displaystyle H\left(x,t;\omega\right)  = \frac{1}{\omega}\mathbb{E}\left[\left.  e^{-\omega \left(t_{d}-t\right)} \right|\mathcal{F}_{t}\right]
\end{equation} 
where, defining $\hat{x} := \ln\left(V_{t} / V_{def}\right)$, we have:
\begin{equation}\label{eq:laplace}
\mathbb{E}\left[\left.  e^{-\omega\left(t_{d}-t\right)} \right|\mathcal{F}_{t}\right]  = \left[   \frac{\gamma(\beta - \eta)}{\eta(\beta- \gamma)} e^{-\beta \hat{x}}  + \frac{\beta(\eta - \gamma)}{\eta(\beta- \gamma)} e^{-\gamma \hat{x}} \right]
\end{equation}   
The values of $\beta$ and $\gamma$ are the two positive roots of the polynomial $\mbox{G}\left(q\right) = \omega$, where:
\begin{equation}\label{eq:polyniamial}
\mbox{G}\left(q\right) = \frac{1}{2}\sigma^2 q^2 -\psi q + \lambda \left( \frac{\eta}{\eta - q} - 1 \right) 
\end{equation}
In equation \eqref{eq:polyniamial}, $\psi$ is the drift of the log-process of the firm, consistent with the risk-adjusted measure $Q$, defined as:
\begin{equation}\label{defaultprobLF}
\psi = r - \frac{1}{2}\sigma^2 - \lambda \left( \frac{\eta}{\eta + 1} - 1 \right)
\end{equation}

\subsection{Pricing Defautable Coupon Bonds}
The value of a defaultable coupon bond at time $t$ with unit nominal value, paying a continuous
coupon rate $b$, with maturity $T$ and a constant recovery of face value $\Upsilon$, is given by:
\begin{equation}\label{Def::B}
B\left(x,t;T\right) =  \mathbb{E}\left[\left. e^{-r\left(T-t\right)} \mathds{1}_{\left\{t_{d}>T\right\}} + \Upsilon e^{-r\left(t_{d}-t\right)} \mathds{1}_{\left\{T\geq t_{d} \geq t\right\}}  + \displaystyle \int_{t}^{T} b e^{-r\left(z-t\right)} \mathds{1}_{\left\{t_{d} > z \right\}} dz  \right|\mathcal{F}_{t}\right]
\end{equation}
\noindent A closed-form solution of the Laplace transform of \eqref{Def::B} is provided in Proposition \ref{Prop::1}. The proof is in Appendix \ref{AppendixA}.

\medskip

\begin{proposition}\label{Prop::1}
Consider B in \eqref{Def::B}, we have that: 
\begin{equation}\label{eq:InvLapl}
B\left(x,t;T\right) = \mathcal{L}^{-1}\left(F\left(x,t;\omega\right)\right)
\end{equation} 
where $\mathcal{L}^{-1}$ is the inverse operator of the Laplace transform of parameter $\omega > 0$. The Laplace transform of the bond price, denoted by  $F$, is defined as:
\begin{equation} \label{eq:bondprice}
\displaystyle F\left(x,t;\omega\right) = H\left(x,t;\omega+r\right)   \left( \displaystyle \frac{b}{r} - \frac{b(\omega + r)}{ r\omega} + \frac{\Upsilon\left(r+\omega\right)}{\omega} - 1  \right)  + \displaystyle \left( \frac{1}{\omega + r} + \frac{b}{r\omega} - \frac{b}{r(\omega + r)} \right)
\end{equation}
where $H$ is defined in \eqref{eq:laplacefirsttimepass}.
\end{proposition}

\subsection{Pricing Credit Default Swaps}

A CDS is a swap financial contract that provides protection to the buyer from adverse events (default) of a company. In particular, the buyer agrees to make continuous payments to the protection seller until the credit event occurs or the contract expires.
The paid amount is called CDS spread. If the default occurs before the maturity of the contract, the protection seller refunds the buyer by the loss-given default. 
In the valuation of such a financial contract, we need to separately address the two components, called legs. The leg whose price depends on the current value of the loss-given default is called the protection leg, whereas the leg concerning the spread is called the premium leg.

Let us consider a CDS written on a bond with unitary face value, initial protection time (current time) $t$, and maturity $T$, and let $t_{d}$ denote the (random) time of default. Moreover, let us assume that the CDS spread, which is denote by $\Pi$, is paid continuously. 
Then, at time $t$, the premium and the protection legs can be defined (see, e.g., \cite{ChenKou2009} and \cite{GunduzUhrigHomburg2014}) as follows: 
\begin{equation}\label{PremiumL}
\mbox{PremiumLeg}\left(x,t;T\right) = \mathbb{E}\left[\left. \displaystyle \Pi \int_{t}^{T} \ e^{-r\left(z-t\right)} \mathds{1}_{\left\{t_{d}\geq z\right\}}dz\right|\mathcal{F}_{t}\right]
\end{equation}
and
\begin{equation}\label{ProtectionL}
\mbox{ProtectionLeg}(x,t;T) = \mathbb{E}\left[\left. \displaystyle e^{-r\left(t_{d}-t\right)}\left(1-\Upsilon\right)\mathds{1}_{\left\{T\geq t_{d} \geq t\right\}}\right|\mathcal{F}_{t}\right]  
\end{equation}
respectively. By equating the premium and protection leg and solving for $\Pi$, we obtain the CDS spread:
\begin{equation}\label{ParSpread_g}
\Pi\left(x,t;T\right) = \frac{\left(1-\Upsilon\right) \mathbb{E}\left[\left.  e^{-r \left(t_{d}-t\right)} \mathds{1}_{\left\{ T\geq t_{d} \geq t \right\}} \right|\mathcal{F}_{t}\right] }{ \mathbb{E}\left[\left. \displaystyle \int_{t}^{T} e^{-r\left(z-t\right)} \mathds{1}_{\left\{t_{d} > z \right\}} dz  \right|\mathcal{F}_{t}\right]} 
\end{equation}

\noindent A closed-form solution of the Laplace Transform of \eqref{ParSpread_g} is given in Proposition \ref{Prop::2}. The proof is in Appendix \ref{AppendixA}.

\medskip

\begin{proposition}\label{Prop::2}
Consider $\Pi$ in \eqref{ParSpread_g}, we have that:
\begin{equation}\label{eq:CDSsperadsSAS}
\Pi\left(x,t;T\right)  = \frac{\mathcal{L}^{-1}\left(F_{ProtectionLeg}\left(x,t;\omega\right)\right)}{\mathcal{L}^{-1}\left(F_{PremiumLeg}\left(x,t;\omega\right)\right)}
\end{equation}
where $\mathcal{L}^{-1}$ is the inverse operator of the Laplace transform of parameter $\omega>0$. The Laplace transform of the premium leg (divided by the CDS spread) is given by:
\begin{equation}\label{eq:premium}
 F_{PremiumLeg}\left(x,t;\omega\right) = \displaystyle H\left(x,t;\omega+r\right)\left(\frac{1}{r} - \frac{\left(r+\omega\right)}{r\omega} \right)  + \frac{1}{r\omega}  - \frac{1}{r\left(r+\omega\right)}
 \end{equation}
Moreover, the Laplace transform of the protection leg is defined as:
\begin{equation}\label{FProtectionLeg}
F_{ProtectionLeg}\left(x,t;\omega\right) =  H\left(x,t;\omega+r\right)\frac{\left(1-\Upsilon\right)\left(r+\omega\right)}{\omega}
\end{equation}
where $H$ is defined in \eqref{eq:laplacefirsttimepass}.
\end{proposition}
In the next subsection, we present the numerical method to compute the inverse Laplace transform.

\subsection{Numerical Inversion of the Laplace transform}
The semi-closed formulas in Propositions \ref{Prop::1} and \ref{Prop::2} for the pricing of defaultable coupon bonds and CDSs, respectively, require the numerical inversion of the Laplace transform. There are several methods we can use to approximate the solution. In this study, we adopt the so-called Gaver-Stehfest algorithm, see, e.g., \cite{Stehfest1970} and \cite{Gaver1966}.
Assume that $F \left( x,t,\omega \right)$ is the Laplace transform of $\Phi (x,t,T)$. According to the Gaver-Stehfest algorithm, the inverse Laplace transform is approximated as follows:
\begin{equation}
    \Phi (x,t,T) \approx \frac{\ln(2)}{T-t}\displaystyle \sum_{k=1}^{2M} \alpha_k^M F \left( x,t,\frac{k \ln(2)}{T-t} \right)
\end{equation}
where
\begin{equation}
    \alpha_k^M = \frac{(-1)^{M+k}}{M!} \sum_{j=\lfloor (k+1)/2 \rfloor}^{\mbox{min}(k,M)}
    j^{M+1} {M\choose j} {2j\choose j} {j\choose (k-j)}
\end{equation}
and $\lfloor \diamond \rfloor$ stands for the integer part of $\diamond$. 

The accuracy of the algorithm depends on the parameter $M$, which has to be chosen correctly in order to avoid rounding errors. We observed via heuristic analysis that it is recommended to choose values of M between $6$ and $8$, as confirmed in the literature, see, e.g., \cite{BallestraPacelliRadi2017}. Indeed, for small values of M, a proper level of accuracy is not always obtained, whereas if the value of $M$ increases, the stability of the algorithm declines as rounding errors can undermine the sum.

In order to validate all the formulas and, therefore, the Gaver-Stehfest algorithm, we study the accuracy of the method by comparing the solutions with an alternative approach. In particular, we computed the survival probability with a finite difference method (FDM). The numerical scheme has been implemented using a sufficiently thin grid, which ensures maximum precision. 
By fixing $M=8$, in all the experiments conducted, we observed that the Gaver-Stehfest algorithm achieves a satisfactory level of accuracy of $10^{-4}$ in comparison with the FDM method.

Moreover, other numerical algorithms to compute the inverse of the Laplace transform have been tested. Specifically, we have also implemented the Bromwich Integral, see, e.g., \cite{CathcartElJahel2006}.  Although it was capable to obtain a proper level of accuracy, the Gaver-Stehfest algorithm was preferred for its computational performance, as it is twice as fast as the Bromwich Integral.

\section{Research Design, Dataset and Methodology}\label{Sec::DatasetAndResearchDesign}
In this section, we outline all the premises before delving into the empirical investigation. First, we explain the research questions and define the hypothesis, distinguishing between the two models we are considering; see Subsection \ref{sec:hypothesis}. Second, we present the dataset employed in the empirical investigation; see Subsection \ref{sec:Data}. Third, we describe the construction of the control variables employed in the analysis; see Subsection \ref{sec:control_variables}. Finally, Subsection \ref{sec:metholodogy} describes the methodology.

\subsection{Background and Hypothesis Definition} \label{sec:hypothesis}
The research design involves studying the model’s capability to capture transition risk. This goal is achieved by calibrating the model daily to market CDS spreads and assessing the implied prices’ relation to transition risk proxies.
First, we consider the empirical evidence regarding the jump-diffusion model, an analysis that aims to confirm the jump-like nature of the transition risk. Then, for comparison, we propose a discussion on a diffusion model. The jump-diffusion model seeks to capture transition risk via its jump component; see Section \ref{Sec::3}. Unlike the diffusion model, it has endorsed more model parameters, resulting in a better capability to fit the term structure of the CDS spreads. In particular, Appendix \ref{Appendix:theoreticalspread} discusses, via sensitivity analysis, theoretical considerations about the model’s
capability to generate different CDS spreads about the firm’s exposure to transition. 
We now enumerate and discuss the four hypotheses we intend to verify.\\ \\
\noindent \textbf{Hp 1.a:} \textit{The jump-diffusion model presents a positive relation between transition risk factor(s) and model CDS spreads.}\\
\noindent Recent studies show the presence of a positive relationship between transition risk indicators and market CDS spreads,
pointing out a quantifiable dependence between each other; see, e.g., \cite{vozian2022climate} and \cite{huij2022carbon}. Our analysis calibrates a pricing model to the market CDS term structure, deriving the implied prices. Therefore, if the model accurately describes the market CDS term structure, we expect that the same analysis performed on jump-diffusion-calibrated prices will produce similar results.
\\ \\
\noindent \textbf{Hp 2.a:} \textit{The jump-diffusion model shows no-significant relation between transition risk factor and calibration errors.}  \\
\noindent Despite the capability of the jump-diffusion model to accurately reproduce market CDS stylized facts, a relevant question we pose is about the pattern of calibration errors. In fact, when we calibrate and estimate the outcomes, the model may underestimate or overestimate the market. 
Being aware of these tendencies, especially in times when transition risk is consistent, is essential guaranteed a prudent adoption of the model. Overall, we expect that calibration errors do not show relevant relations with a proxy of the transition risk, meaning that the jump-diffusion model is able to fit the CDS term structure even in turmoil times.
\\ \\
\textbf{Hp 1.b:} \textit{The diffusion model cannot capture market CDS spreads evidences as the jump-diffusion one.}  
\\ 
\noindent Diffusion models have found extensive applications in describing economic phenomena due to their easy implementation and fast calibration; see, e.g., \cite{Brigo2004CreditDS}. In this paper, we adopt the framework proposed by \cite{merton1974pricing}, a well-known credit risk model commonly considered a benchmark in this kind of topic.    
We aim to verify that this model cannot follow the market CDS stylized fact like the jump-diffusion one, confirming the requirement to design a more complex model to capture transition risk.  
\\ \\
\noindent
\textbf{Hp 2.b:} \textit{The diffusion model shows significant relation between transition risk and calibration errors.}    \\
As in the jump-diffusion framework, we are interested in assessing the model’s capability to reproduce the CDS term-structure, but in this case we expect significant relations between calibration errors and the indicator of transition risk. In fact, it would confirm that a simple diffusion model yields relevant market underestimation or overestimation in acute transition risk period.

\subsection{Data} \label{sec:Data}
We use data belonging to the European markets. In particular, for each firm, we collect the following data:
\begin{enumerate}
    \item CDS spreads.
    \item Firm-specific references and financial data.
    \item Firm-specific GHG data.\footnote{The first definition of greenhouse-gases was stated within the United Nations Framework Convention on Climate Change (1992), an international treaty aiming to combat climate change, by stabilizing emissions in the atmosphere. The formal implementation of these measures arrived with the Kyoto Protocol, in 1997.}
    \item EU carbon price.
    \item Macroeconomic variables. 
\end{enumerate}

\noindent The dataset covers the single-name CDS spreads of companies whose market capitalization is above EUR 250 million. It comprises contracts denominated in Euros and linked to Senior Unsecured Debt. 
The examined period spans five years, starting on January 03, 2017, and ending on December 31, 2021. We collect daily quotes for the entire CDS-term structure, with maturities ranging from 6 months to 30 years. From the initial sample, we do not consider firms where any CDS maturity time series presents missing values or illiquid quotations. To remove outliers, we winsorize the CDS spreads at the 99th percentile.\\ 
\indent We employ firm-specific financial data and credit ratings. We collect daily stock prices and yearly revenues.  
The firm-specific credit ratings are obtained from S\&P’s and Moody’s contributors, choosing Senior Domestic long-term ratings. When the preferred contributors did not provide credit ratings, we don't consider the company in the analysis.\\
\indent Indicators of exposure to transition risk are built employing GHG emissions data and carbon permit allowances future quotations. GHG data comprises firm-specific emissions in each year under study, consisting of Scope 1 and Scope 2 emission categories.
The former gathers emissions that occur directly from sources controlled or owned by an organization, whereas the latter refers to indirect emissions stemming from the company's supply chain. In our analysis, the company is excluded if any firm-specific GHG data is unavailable. The second indicator of transition risk is derived from the carbon price. In particular, we adopt the time series of the futures continuation carbon allowances contract.\footnote{Carbon Price or European Union Emissions Trading System is a scheme to reduce GHG emissions. 
Such policy, which is currently implemented in Europe, establishes a price to carbon emissions by forcing companies to buy contracts (\textit{carbon allowances}) to emit 1 tonne of $\mbox{CO}_2$. As this amount increase, we can expect that the market price of transition risk rises.} Finally, the macroeconomic variable is the Euribor-6m over the 5-year period, which is used as the risk-free interest rate. The final dataset is composed of 84 firms, resulting in 105756 CDS spread. All data has been collected from Refinitiv Eikon database.
\subsection{Control Variables} \label{sec:control_variables}
This study employs several control variables to measure the models' capability to assess transition risk. To select control variables, we build on the recent works of \cite{transitionriskCDSTaschini2021} and \cite{vozian2022climate}, which have addressed similar issues, although by analyzing the impact of the transition risk exclusively from a market perspective. More precisely, our control variables are
\begin{itemize}
    \item[a.]   Stock returns. 
    \item[b.]   Historical stock volatility.
    \item[c.]   Mean Rated Index.
    \item[d.]   Two proxies of the transition risk. 
\end{itemize}
\indent Stock returns are widely consider a proper explanatory variable for the CDS spreads variation, as they are a proxy of the firm's equity value. In the literature, \cite{merton1974pricing} argued that if stock returns increase, the default probability of the company should decrease. Hence, we expect a negative relation between stock returns and CDS spreads. We compute stock returns as $r_{i,t} = \log(E_{i,t}) - \log(E_{i,t-1})$, where $E_{i,t}$ indicates the price of the stock $i$ at day $t$.\\
\indent Historical stock volatility describes the variance of the value of a company. We can expect that, as long as the volatility increases, the perceived credit risk of a company is higher. The relation between CDS spreads and historical volatility should be positive; see, e.g., \cite{GALIL2014271}. We compute the volatility $ \sigma_{i,t}$ (or the variation  $\Delta \sigma_{i,t}$) as the annualized variance of the stock return $r_{i,t}$, on a 255 days rolling window.\\
\indent As the third control variable, we include a factor representing a measure of the market credit risk condition about credit rating. This control variable is the Median Rated Index of \cite{GALIL2014271}, which is defined as the median CDS spread of all the firms belonging to the S\&P's macro categories \{AAA/AA\}, \{A\}, \{BBB\}, \{BB+ or lower\}. We denote this factor as $\mbox{MRI}_{i,t}^m$ (or the variation $\Delta \mbox{MRI}_{i,t}^m)$, where $m$ indicates the maturities of the CDS spread. We mention that, when the S\&P's credit rating was not available, we collect the quotations from Moody’s contributor, being aware of performing the correct conversion. This factor, describing a general condition of credit risk in the market, shows a positive relationship with the CDS spreads; see, e.g., \cite{GALIL2014271}.\\
\indent Finally, we discuss the construction of transition risk proxies. The literature proposes various examples for these factors, mainly based on firm-specific carbon emission data; see, e.g., \cite{vozian2022climate} and \cite{ZHANG2022103286}. This framework lies on studies that have observed a direct relation between a company's exposure to transition risk and emission intensity;
see, e.g., \cite{BOLTON2021517}. Unfortunately, these measures are only available on an annual basis, providing a metric of a company's total emission rate in a given year.
Since our study requires a daily proxy of the transition risk, we cannot use these variables directly, but we need to derive proper daily factors. We use two different indicators of transition risk, which are described below. 
In particular, the first one is borrowed from \cite{transitionriskCDSTaschini2021} and \cite{BlasbergKieselTaschini2022}, whereas the second, as far as we know, is a novel indicator of transition risk. In what follows, in order to lighten the description of the \textit{First Indicator}, we avoid the writing of these references.\\
\indent \textit{First Indicator.}\quad We build a transition risk factor representing the distance between CDS spreads of firms categorized as "green" and "brown". The classification is based on GHG emissions (Scope 1 and Scope 2, normalized by revenues) and credit rating, where low emission-high credit rating belongs the former, while high emitting-low credit rating to the latter.\footnote{Groups are defined based on emission intensity and credit rating. Let $\mbox{ES}_{i,t}$ be the emission intensity and $\mbox{CR}_{i,t}$ be the credit rating, for firm $i=1,...,N$ and time $t = 1,...,T$. The groups are created by dividing firms in quantiles with respect to these metrics, for every $t$. In particular, let $G_{t,j}^m$ be the set of firms at time $t$ where $m$ lies within the $j-th$ quantile. Groups are defined, both for emission intensity and credit rating, as:
\begin{equation}
\begin{array}{lll}
    & \mathcal{G}^{ES}_{t,1} = \{ i | \mbox{ES}_{(i,t)} \leq \mbox{ES}^{1/3}_{t} \}  &  \mathcal{G}^{CR}_{t,1} = \{ i | CR{(i,t)} \leq CR^{1/3}_{t} \}  \\  
    & \mathcal{G}^{ES}_{t,2} = \{ i | \mbox{ES}^{1/3}_{t} < \mbox{ES}_{(i,t)} \leq \mbox{ES}^{2/3}_{t} \}   \quad  & \mathcal{G}^{CR}_{t,2} = \{ i | CR^{1/3}_{t} < CR{(i,t)} \leq CR^{2/3}_{t} \}  \\ 
    & \mathcal{G}^{ES}_{t,3} = \{ i | \mbox{ES}_{(i,t)} > \mbox{ES}^{2/3}_{t} \}   \quad   & \mathcal{G}^{CR}_{t,3} = \{ i | CR{(i,t)} > CR^{2/3}_{t} \}   
\end{array}
\end{equation}
where the apex $1/3$ indicates the subgroup containing the firms with lower emissions $(j=1)$. We note that we compute the groups over the time period under study.\\ Then, combining this preliminary classification, we obtain the final nine groups:
\begin{equation}
    \mathcal{G}_{t;j,k} = \mathcal{G}^{ES}_{t;j} \cap  \mathcal{G}^{CR}_{t;k} \ \ \mbox{where}\ \ j = 1 \dots 3 \, \ \ k = 1 \dots 3
\end{equation}
where green and brown groups are defined as:
\begin{equation}
    \mathcal{G}^{green}_{t} = \{ \mathcal{G}_{t;1 3} \mathcal{G}_{t;1 2} \} \quad \mbox{and} \quad \mathcal{G}^{brown}_{t} = \{ \mathcal{G}_{t;3 1}  \mathcal{G}_{t;3 2} \}
\end{equation}}

\noindent Once the two groups have been defined, we employ two different approaches to define transition risk. The first method consists in estimating this factor as the difference between the median CDS between the two groups:
\begin{equation} \label{eq:TaschiniRegressor_TR_median}
\mbox{TR}_{t}^m = median( \mbox{CDS}^{brown}_t) - median( \mbox{CDS}^{green}_t )   
\end{equation}
\begin{figure}[H] 
	\begin{centering}
		\begin{tabular}{cc}
		        \includegraphics[scale=0.5]{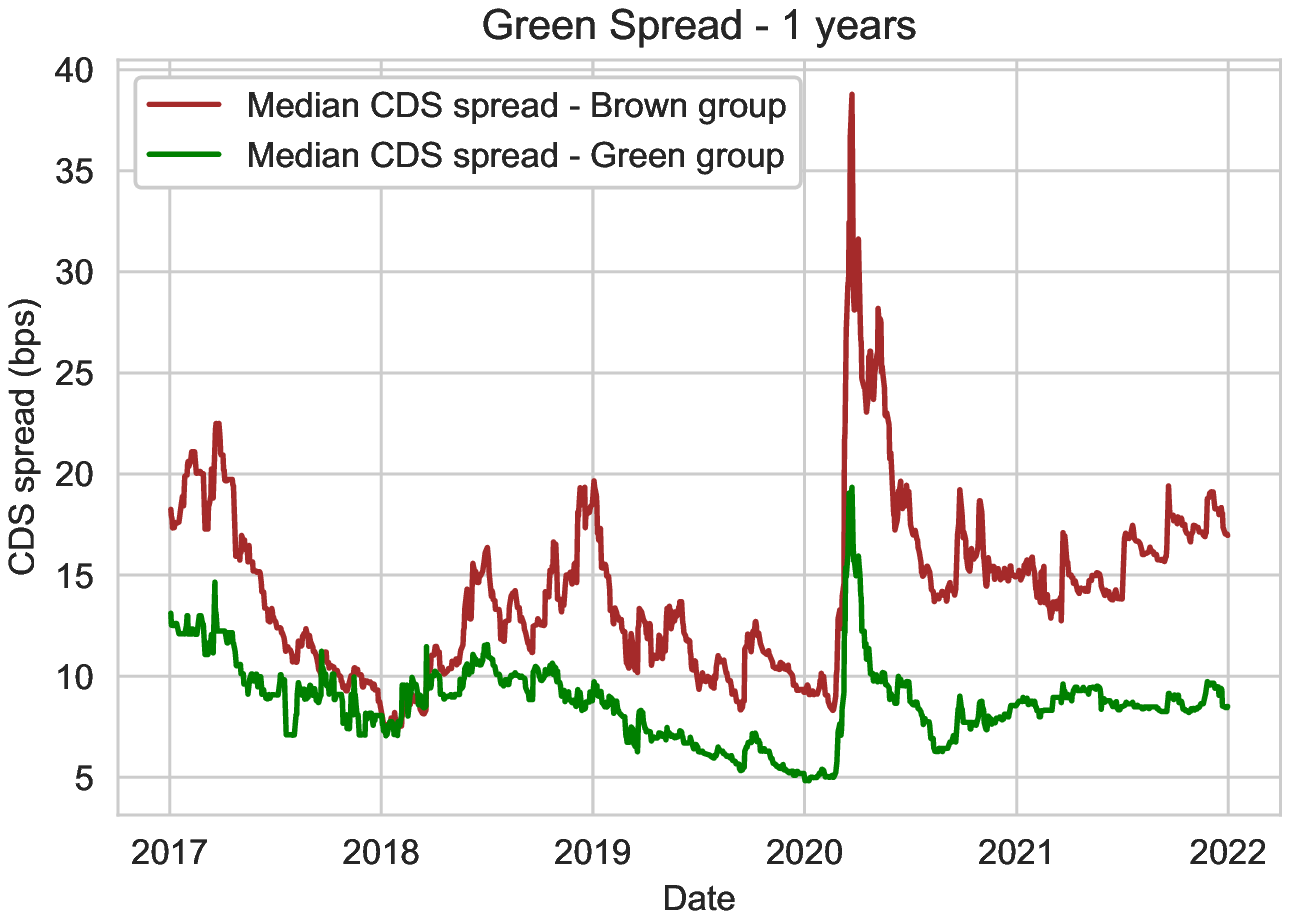} & \includegraphics[scale=0.5]{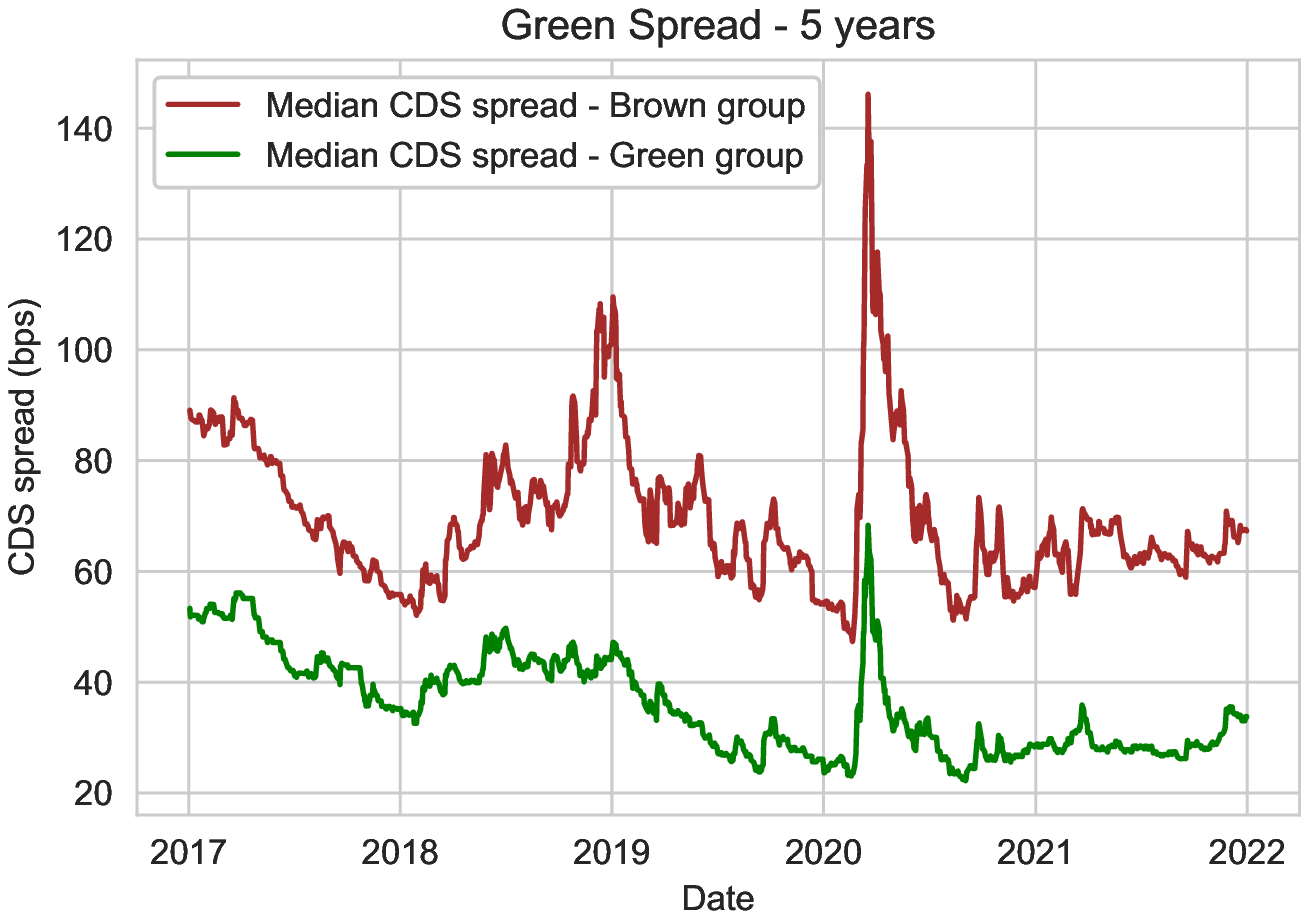}  \tabularnewline[-1ex]
				 \tabularnewline[1ex]
			    \includegraphics[scale=0.5]{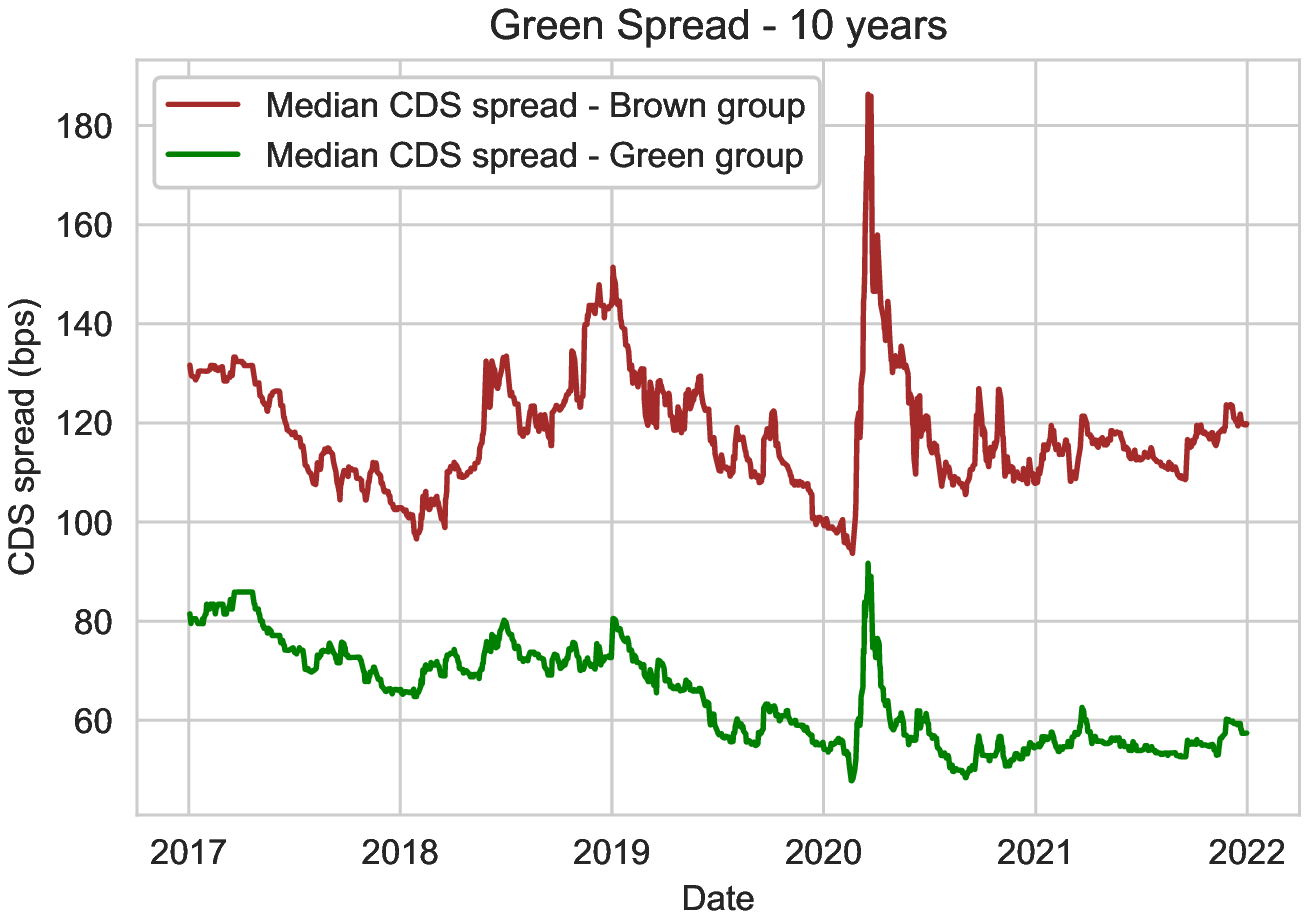} & \includegraphics[scale=0.5]{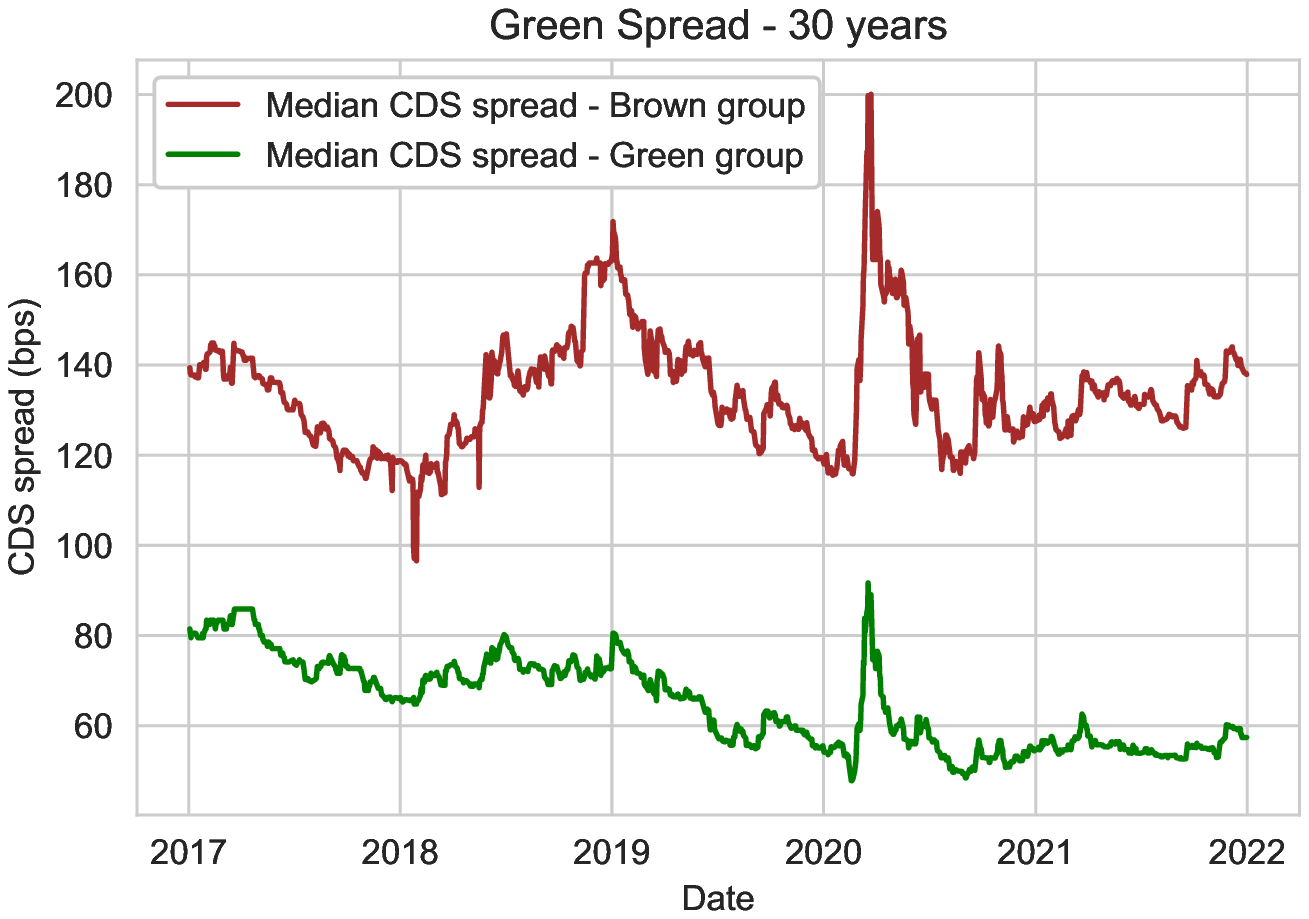} \tabularnewline[-1ex]
		\end{tabular}
	\caption{This figure reports the median CDS spread for companies belonging to green and brown groups at different maturities (1 year, 5 years, 10 years, and 30 years), providing a visual representation of the green spread estimated from the market. }\label{fig:GreenSpread}
	\par\end{centering}
	\centering{}
\end{figure}
Figure \ref{fig:GreenSpread} shows the market CDS spread constructed by the metric presented in \eqref{eq:TaschiniRegressor_TR_median}. We can notice how the difference between green and brown firms is significant and, in particular, the spread increases along the maturities, indicating that the market price of transition risk is higher on the long term. The spike we can spot at the beginning of the year 2020, which is due to the COVID-19 pandemic financial turmoil, affects as expected both green and brown companies.

The second method, instead, is based on the derivation of the Wasserstein distance between two empirical cumulative distribution functions (ECDF). In particular, we compute the ECDF of the CDS spreads of the companies in the green and brown group, where $m$ indicates the maturities, as defined:
\begin{equation}
\begin{array}{cc}
     \hat{F}^{m}_t (x)  = \displaystyle \frac{1}{ \displaystyle | \mathcal{G}^{green}_{t} |} \displaystyle \sum_{i \in \mathcal{G}^{green}_{t}} \mathds{1}_{\left( CDS_{i,t}^m \leq x  \right)}   \quad \mbox{and}\qquad
     \hat{G}^m_t (x) =  \frac{1}{ \displaystyle | \mathcal{G}^{brown}_{t} |} \displaystyle \sum_{i \in \mathcal{G}^{brown}_{t}} \mathds{1}_{\left( CDS_{i,t}^m \leq x  \right)}.
\end{array}
\end{equation}
\noindent At this point, it is possible to measure the evolution of the transition risk over time.
To do this, we adopt the first-order Wasserstein distance, computed between the empirical cumulative distribution functions, defined as:
\begin{equation} \label{eq:TransitionRIskTaschini_Weiss}
   \mbox{TR}_t^{m} = \displaystyle \int_0^1 \left| \left( \hat{F}^{m}_t \right)^{-1}(u) -  \left( \hat{G}^{m}_t \right)^{-1}(u)   \right| du
\end{equation}
Intuitively, as long as the perceived level of transition risk in the market rises, the relative distance between green and brown CDS spreads increases, resulting in a larger measure of transition risk. This metric, which is used to gauge the distance between two distributions, has found large adoption in the literature; see, e.g., \cite{HOSSEININODEH2023103735}, \cite{GOFFARD2021350}, and  
\cite{FENG2023107166}.
\\
\indent \textit{Second Indicator.}\quad We introduce a novel approach to model transition risk, based on the carbon price. 
The basic idea is that significant perturbations in the time-series of carbon price should be related to the reduction of emission limits or the design of tighter green laws, i.e. drivers of transition risk.
\noindent
In particular, we can expect that brown companies should be more receptive than green companies to these dynamics, resulting in an impact on their CDS spreads. We denote the carbon price factor as $\mbox{CP}_t$ (or the variation $\Delta \mbox{CP}_t$). \\ \noindent In Figure \ref{fig:CarbonPrice}, we report the spot price for carbon allowances in Europe (January 03, 2017 - December 31, 2021). Our approach lies on a shared thought in the literature, whereby the carbon price is consider a driver of transition risk. 
\begin{figure}[H]
    \centering
    \includegraphics[scale=0.65]{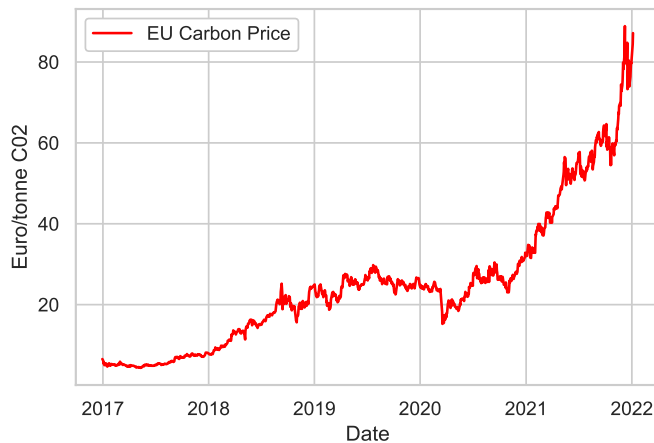}
    \caption{Price of carbon allowances in Europe.}
    \label{fig:CarbonPrice}
\end{figure}
Nevertheless, recent studies introduced this metric from a theoretical perspective only, as outlined in the following. \cite{Belloni_carbon_price} studies the impact of transition risks deriving from climate change mitigation policies on the ﬁnancial system, by assessing the consequences of carbon price rising via a banking sector contagion model.
\cite{le2022corporate} investigates how different future transition scenarios can affect the price of debt securities, assuming different theoretical carbon prices based on enactment of more stringent green policies.
So far, as we know, no studies have adopted the Carbon Price as an indicator of transition risk, in particular regarding the CDS spread pricing framework. Descriptive statistics of market CDS spreads and control variables are presented in Appendix \ref{Appendix:statistics}.

Before introducing the methodology, we mention further studies which adopted different control variables to investigate the change in CDS spreads compared to transition risk. 
In \cite{vozian2022climate}, a fine collection of market data has been used to regress several financial variables on the CDS spreads. In addition to GHG emissions proxy, the work proposes other indicators of the effort of the firm’s management to face green transition, which are defined based on internal policies, reduce-emission target setting and internal carbon pricing.  \\
In \cite{AskBertCDS}, the authors formulate a study on the effects of disclosing company’s 
 climate risks on the CDS spread. These measures are created by BERT, a Natural-Language-Processing algorithm which is able to extract quantifiable information from text data. In particular, this algorithm has been used to measure the firm exposure to transition and physical risk, analyzing a section of the 10-K filings document, a financial report where companies are obliged to disclose financial condition. Although these approaches are interesting they are beyond the scope of our work, whereby we only refers to GHG-related factors.

\subsection{Methodology} \label{sec:metholodogy}
Based on the framework of \cite{transitionriskCDSTaschini2021}, the research investigation is performed by applying a panel quantile regression, assessing the relations along the conditional distribution of the dependent variable, i.e. the CDS term-structure.
\\
Many authors en-light the advantages of this technique compared to standard linear regression. Indeed, quantile regression can mitigate the presence of outliers and non-normal errors, relaxing the assumptions underlying the application of standard linear regression; see, e.g., \cite{Quantile_REGR}. Let $y_{i,t}$ be the dependent variable, where $i = 1,\dots,N$ and $t = 1,\dots,T$.\footnote{From now on, $T$ refers to the number of time-based (daily) observations that we are considering in our empirical analysis.} The quantile regression is defined as: 
\begin{equation}
    \mbox{Q}_{y_{i,t}}( \tau | \boldsymbol{x_{i,t}} )  =  \alpha_{i,t} + \boldsymbol{x_{i,t}} \beta_{\tau} + \epsilon_{i,t}
\end{equation}
where $\tau \in (0,1)$, Q is the conditional quantile of $y_{i,t}$ with respect to $\boldsymbol{x_{i,t}}$ ($m$-dimensional regressor vector), and $\alpha_{i,t}$ indicates the fixed-effects. 
The estimation of the coefficients is performed in two steps; see, e.g.,  \cite{transitionriskCDSTaschini2021} and \cite{QuantileZhang}. First, we perform firm-specific quantile regressions to derive the firm-related fixed effects $\alpha_{i,t}$. This is achieved by:
\begin{equation} \label{eq:fixed_effect_coeff}
    ( \Tilde{\alpha}_{i,\tau} , \Tilde{\beta}_{i,\tau}) =  \underset{a \in \mathbb{A},b \in \mathbb{B}_{\tau}}{\mbox{argmin}} \sum_{t=1}^T  \rho_{\tau} \ (y_{i,t} - a - \boldsymbol{x_{i,t}}b ) 
\end{equation}
where $\mathbb{A} \subseteq \mathds{R}$, $\mathbb{B}_{\tau} \subseteq \mathds{R}^m$ and $\rho_{\tau}(u) = u(\tau - \mathds{1}_{ \{u<0 \} }  )$ is the quantile error function.
Finally, the regression coefficients are estimated by:
\begin{equation} \label{eq:complete_regression}
     \hat{\beta}_{\tau} =  \underset{b \in \mathbb{B}_{\tau}}{\mbox{argmin}} \sum_{i=1}^N\sum_{t=1}^T  \rho_{\tau} \ (y_{i,t} - \alpha_{i,\tau} - \boldsymbol{x_{i,t}}b ) 
\end{equation}
\section{Empirical Investigation and Results}\label{sec:empirical}
In this section, we describe the empirical investigation performed to test the hypothesis introduced in the previous section. Subection~\ref{sec:calibration} explains the calibration procedure, distinguishing between the jump-diffusion and the diffusion model. Subections \ref{panel_regression_model} and \ref{panel_regression_error} states the empirical investigation, discussing the two different experiments. Subection~\ref{sec:Results} presents the results.
\subsection{Calibration Procedure} \label{sec:calibration}
The calibration procedure, which consists in finding the optimal model parameters to fit the market CDS term-structure, is performed on both the models we are considering. In particular, we perform 1-day calibrations over the 5-year period under study on the complete CDS term-structure, for time-to-maturities 
$\hat{t}_1 = 0.5,\ \hat{t}_2 = 1,\ \hat{t}_3 = 2,\ \hat{t}_4 = 3,\ \hat{t}_5 = 4,\ \hat{t}_6 = 7,\ \hat{t}_8 = 10,\ \hat{t}_9 = 20,\ \hat{t}_{10} = 30 \ \mbox{years}$.
The optimization function is the same for both models, but the number of calibrated parameters is different. The simpler model, which does not contain jumps, is driven by only $\frac{V_{t}}{V_{def}}$ and $\sigma$ parameters, meanwhile, the jump-diffusion one requires also $\lambda$ and $\eta$. The formula to price a CDS under the diffusion model is outlined in Appendix \ref{Appendix:CDS_diffusion}. For plain notation, let $\Omega_t$ be the set of calibrated parameters at day $t$. We can define for the two models, where \textit{jd} is \textit{jump-diffusion} and \textit{d} is \textit{diffusion}, the sets of parameters: 
\begin{equation}
    \Omega_t^{jd} := \left\{ \frac{V_{t}}{V_{def}},\sigma , \lambda , \eta \right\}_t    \quad \mbox{and} \quad                           \Omega_t^{d} := \left\{  \frac{V_{t}}{V_{def}},\sigma      \right\}_t
\end{equation}
 Following \cite{BallestraPacelliRadi2020}, we minimize the mean absolute percentage error (MAPE):\footnote{An alternative error measure that is commonly adopted is the root mean square error. Nevertheless, in fitting the term structures of CDS spreads a relative error measure as the MAPE enables us to avoid over-weighting the mispricing in fitting long-term CDS spreads (the highest ones) and under-weighting the mispricing in fitting short-term CDS spreads (the lowest ones).}
\begin{equation}\label{eq:mape}
MAPE \left( t,\Omega_t^{ \{ k \}  } \right)=\frac{1}{10}\sum\limits_{i=1}^{10}\left|\frac{\Theta_t^{market}(\hat{t}_i)- \Theta_t^{model}(\hat{t}_i,\Omega_t^{ \{ k \}  })}{\Theta_t^{market}(\hat{t}_i)}\right|.
\end{equation}
where $k = \{j,jd \}$, $\Theta_t^k(\hat{t}_i,\Omega_t^{ \{ k \}  }) = \Pi\left(x,t;\hat{t}_i\right)$ if $k = jd$  and $\Theta_t^k(t,\hat{t}_i) = \Lambda\left(x,t;\hat{t}_i\right)$ if $k = d$ .  The recovery rate $\Upsilon$ is set to 0.6, for both models. These values has been chosen following common lines presented in the literature; see, e.g., \cite{JANKOWITSCH20081269}.
\subsection{CDS Model-Implied Panel Quantile Regression} \label{panel_regression_model}
The first panel quantile regression investigates \textbf{Hp 1.a} and \textbf{Hp 1.b}.  It is performed by computing the first-difference of the CDS spreads generated by the calibrated models, observing the relation with a proxy of the transition risk, in order to examine whether the model is able to reproduce the CDS market stylized facts. The quantile regression is defined as:
\begin{equation} \label{eq:Firstexperiment}
    \mbox{Q}_{\Delta \mathcal{\mbox{CDS}}_{i,t}^{model}}( \tau | \boldsymbol{x_{i,t}} )  =  \alpha_{i,\tau} + \beta_{1,\tau} r_{i,t} +  \beta_{2,\tau} \Delta \sigma_{i,t} +  \beta_{3,\tau} \Delta \mbox{MRI}_{i,t} +  \beta_{4,\tau} \Delta \mbox{TR}_{t}+\epsilon_{i,t}
\end{equation}
where $i=1,...,N$ indicates the firm while $t=1,...,T$ the time. The regressors comprises the stock return $r_{i,t}$, the volatility variation $\Delta \sigma_{i,t}$, the common factor $\Delta \mbox{MRI}_{i,t}$ and the changes transition risk $\Delta \mbox{TR}_{t}$. The fixed-effect $\alpha_{i,\tau}$, which are firm specific and change with respect to the considered quantile, are computed by \eqref{eq:fixed_effect_coeff}. The dependent variable is defined as:
\begin{equation}\label{eq:error_diff}
\begin{array}{cc}
    \Delta \mathcal{\mbox{CDS}}_{i,t}^{model} =  \mathcal{\mbox{CDS}}_{i,t}^{model} - \mathcal{\mbox{CDS}}_{i,t-1}^{model}
\end{array}
\end{equation}
where $\mathcal{\mbox{CDS}}_{i,t}^{model}$ denotes the model-implied CDS computed at day $t$ for the firm $i$.
\subsubsection{Calibration Errors Panel Quantile Regression} \label{panel_regression_error}
The second panel quantile regression investigates \textbf{Hp 2.a} and \textbf{Hp 2.b}. It is performed by adopting as dependent variable the difference between the model CDS spread and market CDS spread at day $t$ (see \eqref{eq:error}). In this case, we want to gauge whether the model underestimate or overestimate market CDS spreads, to determine if there is a significant relation between calibration errors and a proxy of the transition risk. The quantile regression is defined as:
\begin{equation}  \label{eq:Secondexperiment}
    \mbox{Q}_{\mathcal{E}_{i,t} } ( \tau | \boldsymbol{x_{i,t}} )   =  \gamma_{i,\tau} + \beta_{1,\tau} r_{i,t} +  \beta_{2,\tau}  \sigma_{i,t}  +  \beta_{3,\tau}  \mbox{TR}_{t}+\epsilon_{i,t}
\end{equation}
where $i=1,...,N$ indicates the firm while $t=1,...,T$ the time. In this case, we do not consider as regressors the variation of the control variables, instead they are defined as the value at day $t$. In particular, the regressors employed are the stock return $r_{i,t}$, the volatility $\sigma_{i,t}$ and the proxy of the transition risk $\mbox{TR}_{t}$. Similarly to the previous case, the firm-specific fixed effects $\gamma_{i,\tau}$ are estimated via \eqref{eq:fixed_effect_coeff}.
The dependent variable, computing the calibration error at day $t$ between $model$ and $market$ CDS spread, is defined as:
\begin{equation} \label{eq:error}
\begin{array}{cc}
     \mathcal{E}_{i,t} =  \mbox{CDS}_{i,t}^{model} - \mbox{CDS}_{i,t}^{market}  
\end{array}
\end{equation}

The panel quantile regressions defined in Subsections \ref{panel_regression_model} and \ref{panel_regression_error} are both performed over the entire CDS-term structure, with maturities $\hat{t}_i \in \{0.5, 1, 2, 3, 4, 5, 7, 10, 20, 30\}$ years and for quantiles $\tau \in \{ 0.1, \dots, 0.9  \}$. Results are presented in Appendix \ref{Appendix:Results} and discussed in the next section.\footnote{From now on, the discussion will be related to both the experiments (\ref{eq:Firstexperiment}, \ref{eq:Secondexperiment}); we avoided the introduction of other subsections for the sake of  readability.}\\
We mention that the panel quantile regressions have been tested to assess whether regressors present multicollinearity. This is achieved by performing the Variance Inflation Factor test (VIF, henceforth), which produces a measure of the correlation between the variables; see, e.g., \cite{VIFlessthan2.5} and \cite{VIFlessthan5}. In particular, if the outcome is greater than a given value one can choose arbitrarily, the regression is considered affected by multicollinearty. Following standard approaches in the literature, we set the threshold level to 5. We perform multicollinearity test for both the estimation of fixed effects while doing firm-specific quantile regressions \eqref{eq:fixed_effect_coeff}, and  for the second quantile regression \eqref{eq:complete_regression}. 
In the first case, considering the regressors as the first-difference of the control variables \eqref{eq:Firstexperiment}, all factors presented VIF values less than 2, thus no adjustments have been necessary.  Meanwhile, in the second  panel quantile regression \eqref{eq:Secondexperiment}, when regressors are considered as the value of the control variable at day $t$, the VIF test founds strong collinearity (above 15) between the proxy of the transition risk built by the Wasserstein distance (or derived by \eqref{eq:TaschiniRegressor_TR_median}) and the MRI factor. For this reason, we decided not to include the MRI factor among the regressors in the absolute-value case.

\indent To verify if the calibration errors are statistically significant, we have performd a $t$-test\footnote{We perform a $t$-test at 5\% confidence level.} on the difference between the calibration errors of the jump-diffusion model and the calibration errors of the diffusion model. 
The test has been carried out on several firms and for various CDS spread maturities. In all cases, the null hypothesis has been verified, meaning that there is no statistically significant difference between the models calibration errors.
\subsection{Results} \label{sec:Results}
The results of the empirical investigation are shown in the tables in Appendix \ref{Appendix:Results} .
The outcomes regarding the first hypothesis are reported in Table \ref{tab:jump-diffusion-deltaCDSmodel}.
First, we can notice a negative relationship between model-implied $\Delta \mbox{CDS}$ spreads and the stock returns, as we expect according to theoretical and empirical findings. 
In fact, this result is expected from structural models theory (see, e.g., \cite{merton1974pricing}), which states that as stock returns increase, the probability of default decreases, as well as the CDS spreads. 
Empirical evidences supporting this fact come from \cite{transitionriskCDSTaschini2021}, \cite{GALIL2014271}, and \cite{Pires2015TheED}. 
Second, the relationship with the volatility variation ($\Delta \sigma$) varies accordingly to the quantile considered, moving from negative values at smaller quantiles to positive values at higher quantiles.
This partially contrasts the theoretical background, which states that a positive variation in the CDS spreads should be related to an increment in the volatility, denoting an increase in the default probability. Nevertheless, there are empirical studies that confirm our findings, assigning negative values to volatility in downward periods and positive values in upward periods; see \cite{transitionriskCDSTaschini2021}, \cite{KOUTMOS201919}.
Third, we appreciate a positive relation with the MRI factor, indicating that the calibrated model prices are following the  general trends of the CDS market, as outlined empirically by \cite{GALIL2014271}.  
For our purposes, the most interesting is the estimated relationship with the proxy of the transition risk. We can notice how, regardless of the quantile or the maturity considered, the relationship is positive. Again, this confirms that the jump-diffusion is able to properly depict the stylized facts of the CDS market.
The regressions are performed adopting the transition risk factor defined in \eqref{eq:TransitionRIskTaschini_Weiss} only, 
as no major difference was found compared to the variation of the first indicator computed by \eqref{eq:TaschiniRegressor_TR_median}.
We also note that, for all the maturities and the quantiles examined, the regressions are significant. 
\\
The analysis regarding the third indicator of the transition risk, defined by the \textit{Carbon Price} ($\mbox{CP}_t$), is not proposed in the empirical investigation, since it shows a negative correlation with indicators of transition risk defined in \eqref{eq:TaschiniRegressor_TR_median} and \eqref{eq:TransitionRIskTaschini_Weiss}. This is reported in Appendix \ref{Appendix:statistics}, Tables \ref{tab:correlation_CP_weiss} and \ref{tab:correlation_CP_med}. Further analysis will be performed to assess whether the carbon price can be considered a suitable proxy for the transition risk.

Table \ref{tab:jump-diffusion-errorcalib} presents the results regarding the relationship between calibration errors and the control variables obtained by the jump-diffusion model, presented in \eqref{eq:Secondexperiment}.
In this case, we are mainly interested in the relation that the calibration error assumes with the transition risk. As we can see, there is no significant relation between between the error and the transition risk. This confirms our hypothesis, meaning that the model is able to capture the transition risk in the calibration procedure.
Overall, the insights we can glean from these two experiments provide a first confirmation that transition risk assumes a \textit{jump-like} nature and that the jump-diffusion model can properly be calibrated even in this new market scenario.

Tables \ref{tab:diffusion-deltaCDSmodel} and \ref{tab:diffusion-errorcalib} provide the outcomes of the same analysis, whereby the calibrated model is the diffusion one.
In Table \ref{tab:diffusion-deltaCDSmodel}, we can notice that the model cannot properly reproduce the market stylized facts of the CDS market: The relationship of the dependent variable ($\Delta \mbox{CDS}$) with the stock returns is mostly positive, in contrast with theoretical and empirical expectations. Moreover, we find high $p$-values for several regressions, meaning that they are not always significant.
The estimated relationship with the volatility is more aligned with the jump-diffusion model predictions, although we obtain negative and positive values along the CDS maturities. 
Regarding the transition risk, again, the estimation is not in line with market expectations. For short and medium maturities (from 6 months to 5 years), we can observe negative relations, in contrast with market and jump-diffusion model evidences. Furthermore, for certain quantiles, the regressions are not significant, assuming very high $p$-values. Therefore, we can argue that the diffusion model is not able to reproduce the stylized facts of the CDS market.

The investigation of the last hypothesis is described in Table \ref{tab:diffusion-errorcalib}. In contrast with expectations, we can observe that there is no significant relationship with the transition risk, similarly to the jump-diffusion model.
This might suggest that the diffusion model is able to capture the transition risk but, however, this result is not relevant. In fact, if we couple this result with the observations raised by Table \ref{tab:diffusion-deltaCDSmodel}, showing that it does not depict the stylized facts of the CDS market, we can argue that the diffusion model is not able to price transition risk.
\section{Conclusions}\label{Sec::6}
Developing pricing models that can properly capture transition risk from the market is becoming a relevant issue in quantitative finance. As climate change worsens, regulators' efforts to reduce GHG emissions and mitigate environmental damages can have severe consequences on the economies, forcing companies to adjust their business models and threatening the stability of the financial system. Hence, the financial sector must be equipped with state-of-the-art pricing models that can adequately manage climate risk factors and classic financial risks (\textit{market, credit, liquidity}), by which it will be possible to examine the impact of transition risk on a company's creditworthiness, gauging the effect of green policies on the financial stability of the firm.

In this study, we aim to assess whether a structural credit-risk jump-diffusion model is able to capture transition risk as priced in the market, where the firm dynamic is affected by downward jumps, describing the impact on green policies on the firm valuation. This is achieved by performing daily calibrations of the model to the market CDS term-structure of 84 European firms, from January 03, 2017 to December 31, 2021, and assessing the relationships between implied prices and a proxy of the transition risk, obtained via panel quantile regression. The same analysis is performed on a diffusion model, to compare the model capabilities.
Results are remarkable: The jump-diffusion model does not only outperform the diffusion model in a proper depiction of the market stylized facts of the CDS market, but also it shows no significant relationships between calibration errors and the proxy of transition risk, confirming the \textit{jump-like} nature of the transition risk. 
As future research, relevant topics regard the implementation of more complex models that accommodate additional climate and macroeconomic risk factors.
Currently, we are working on the same structural credit-risk model where the interest rate is stochastic, in order to further enhance the model's capabilities to calibrate the market CDS term-structure.

\bigskip
\bigskip

\textbf{Acknowledgments:} The Authors acknowledge Fabrizio Lillo for comments on a early version of this paper. 

\bigskip

\textbf{Credit Author Statement:} Giulia Livieri: Formal analysis, Investigation, Methodology, Review \& editing. Davide Radi: Formal analysis, Investigation, Methodology, Review \& editing. Elia Smaniotto: Formal analysis, Investigation, Methodology, Data curation, Writing. 

\bigskip
\textbf{Conflict of Interest Statement:}
The authors declare that they do not have any conflict of interest.

\bigskip

\textbf{Ethical Approval:}
This article does not contain any studies with human participants or animals performed by any of the authors.

\bigskip

\textbf{Data Availability:} The data employed in this study is owned by Refinitiv Eikon. The authors accessed the data with their institution’s membership. The authors had no special access privileges to the data others would not have. The data underlying the results presented in the study are available from the Refinitiv Eikon Database website (https://eikon.thomsonreuters.com/index.html).

\bigskip

\begin{appendices}
\section{Technical Results}\label{AppendixA}

\begin{proof}[Proof of Proposition \ref{Prop::1}]
Consider the bond price $B\left(x,t;T\right)$ in \eqref{Def::B} and denote by $F\left(x,t;\omega\right) = \mathcal{L}\left(B\left(x,t;T\right)\right)\left(\omega\right)$ its Laplace transform. Then,
\begin{equation}\label{LaplaceTofB}
F \left(x,t;\omega\right) =  \displaystyle \int_{t}^{\infty} e^{-\omega\left(T-t\right)}\mathbb{E}\left[\left. e^{-r\left(T-t\right)}  \mathds{1}_{\left\{t_{d} > T\right\}} + \Upsilon  e^{-r\left(t_{d} - t\right)}  \mathds{1}_{\left\{T \geq t_d \geq t \right\}}  + \int_{t}^{T}  b  e^{-r\left(z-t\right)} \mathds{1}_{\left\{t_{d} > z \right\}} dz   \right|\mathcal{F}_{t}\right]  dT
\end{equation}
By Fubini's theorem and straightforward considerations, formula \eqref{LaplaceTofB} can be rewritten as follows
\begin{equation}\label{LaplaceTofB2}
\begin{array}{lll}
F \left(x,t;\omega\right) & = &\displaystyle \int_{t}^{\infty} e^{-\left(\omega+r\right)\left(T-t\right)}\mathbb{E}\left[\left.\mathds{1}_{\left\{t_{d}>T\right\}}\right|\mathcal{F}_{t}\right]dT +  \displaystyle \int_{t}^{\infty} \Upsilon  e^{-\omega\left(T-t\right)} \mathbb{E}\left[\left. e^{-r\left(t_{d}-t\right)}\mathds{1}_{ \{T \geq t_d \geq t \} }\right|\mathcal{F}_{t}\right]  dT \\
\\
 & & + \displaystyle \int_{t}^{\infty} e^{-\omega\left(T-t\right)}  \mathbb{E}\left[\left.\int_{t}^{T}  b e^{-r\left(z-t\right)} \mathds{1}_{ \{ t_d >z \} } dz \right|\mathcal{F}_{t}\right]  dT
\end{array}
\end{equation}
The following are the computations to obtain the final Laplace transform, which will be splitted in order to compute the three addendum of \eqref{LaplaceTofB2} separately. The \textit{first addendum} is defined as:
\begin{equation} \label{eq:laplacebond}
\begin{array}{lll}
\displaystyle \int_{t}^{\infty} e^{-\left(\omega+r\right)\left(T-t\right)}\mathbb{E}\left[\left.\mathds{1}_{\left\{t_{d}>T\right\}}\right|\mathcal{F}_{t}\right]dT  & = &
    \displaystyle \mathbb{E}\left[\left. \int_{t}^{t_d} e^{-\left(\omega+r\right)\left(T-t\right)} dT \right|\mathcal{F}_{t}\right]\\
    \\
    & =&  \displaystyle \frac{1}{\omega + r} \displaystyle \left[  1- \mathbb{E}\left[\left.e^{-\left(r+\omega\right)\left(t_{d}-t\right)}\right|\mathcal{F}_{t}\right] \right]
\end{array}
\end{equation}
The \textit{second addendum} is defined as:
\begin{equation} \label{eq:laplacerecovery}
\begin{array}{lll}
\displaystyle  \Upsilon \int_{0}^{\infty} e^{-\omega\left(T-t\right)} \mathbb{E}\left[\left. e^{-r\left(t_{d}-t\right)}
\mathds{1}_{\left\{T\geq t_d\geq 0\right\}}\right|\mathcal{F}_{t}\right] dT &= &  \Upsilon  \displaystyle  \mathbb{E}\left[\left. \int_{t_d}^{\infty}  e^{-\omega\left(T-t\right)-r\left(t_{d}-t\right)} dT \right|\mathcal{F}_{t}\right]\\
\\
& =&  \displaystyle \frac{\Upsilon}{\omega} \mathbb{E}\left[\left.e^{-\left(r+\omega\right)\left(t_{d}-t\right)}\right|\mathcal{F}_{t}\right] 
\end{array}
\end{equation}
Before providing the final form of \textit{third addendum}, we note that:
\begin{equation}\label{eq:pre_proofcoupon}
\begin{array}{lll}
\mathbb{E}\left[\left. \displaystyle \int_{t}^{T} b e^{-r\left(z-t\right)} \mathds{1}_{ \{ t_d >z \} } dz \right|\mathcal{F}_{t}\right] & = &  b \mathbb{E}\left[\left. \displaystyle \int_{t}^{T \wedge t_d}   \ e^{-r\left(z-t\right)}   dz \right|\mathcal{F}_{t}\right] \\
\\
& = & \displaystyle \frac{b}{r} \mathbb{E}\left[\left. 1 - e^{-r\left(\left(T \wedge t_{d}\right)-t\right)}\right|\mathcal{F}_{t}\right] \\ 
\\
& = &  \displaystyle \frac{b}{r} \mathbb{E}\left[\left. 1 - e^{-r\left(t_{d}-t\right)} \mathds{1}_{ \{t_d < T \} }\right|\mathcal{F}_{t}\right] - \frac{b}{r} e^{-r\left(T-t\right)} \mathbb{E}\left[\left. \mathds{1}_{ \{ t_d >T \} }  \right|\mathcal{F}_{t}\right] 
\end{array}
\end{equation}
where $\wedge$ is the min operator. Employing \eqref{eq:pre_proofcoupon}, the \textit{third addendum} will be:
\begin{equation} \label{eq:coupon_proof}
\begin{array}{lll}
\displaystyle \int_{t}^{\infty} e^{-\omega\left(T-t\right)} \mathbb{E}\left[\left. \displaystyle \int_{t}^{T} b e^{-r\left(z-t\right)} \mathds{1}_{ \{ t_d >z \} } dz \right|\mathcal{F}_{t}\right] dT & = & \displaystyle\int_{t}^{\infty} e^{-\omega\left(T-t\right)} \displaystyle \frac{b}{r} \mathbb{E}\left[\left. 1 - e^{-r\left(t_{d}-t\right)} \mathds{1}_{ \{t_d < T \} }\right|\mathcal{F}_{t}\right] dT\\
\\
& & -\displaystyle \int_{t}^{\infty} e^{-\omega\left(T-t\right)} \ \displaystyle   \frac{b}{r} e^{-r\left(T-t\right)} \mathbb{E}\left[\left. \mathds{1}_{ \{ t_d >T \} }  \right|\mathcal{F}_{t}\right]   dT
\end{array}
\end{equation}
Then, computing separately the two parts of the right-hand side of \eqref{eq:coupon_proof}, we have:
\begin{equation}
\begin{array}{lll}
    \displaystyle \int_{t}^{\infty} e^{-\omega\left(T-t\right)} \left(\displaystyle \frac{b}{r} \mathbb{E}\left[\left. 1 - e^{-r\left(t_{d}-t\right)} \mathds{1}_{ \{t_d < T \} }\right|\mathcal{F}_{t}\right]  \right) dT & = & \displaystyle \frac{b}{r} \int_{t}^{\infty}  e^{-\omega\left(T-t\right)} dT - \frac{b}{r} \mathbb{E}\left[\left. 
    \displaystyle \int_{t_{d}}^{\infty}e^{-r\left(t_{d}-t\right)-\omega\left(T-t\right)}dT \right|\mathcal{F}_{t}\right] \\
    \\
   & = & \displaystyle \frac{b}{r \omega} - \frac{b}{r \omega} \mathbb{E}\left[\left.  e^{-\left(r+\omega\right) \left(t_{d}-t\right)}\right|\mathcal{F}_{t}\right] 
\end{array}
\end{equation}
and
\begin{equation}
\begin{array}{lll}
- \displaystyle \int_{0}^{\infty} e^{-\left(\omega+r\right)\left(T-t\right)}  \frac{b}{r} \mathbb{E}\left[\left. \mathds{1}_{ \{ t_d >T \} }  \right|\mathcal{F}_{t}\right]   dT & = & - \displaystyle \frac{b}{r} \mathbb{E}\left[\left.\int_{0}^{t_{d}} e^{-\left(\omega+r\right)\left(T-t\right)}dT \right|\mathcal{F}_{t}\right]\\
\\
& = & \displaystyle \frac{b}{r\left(r+\omega\right)}\mathbb{E}\left[\left.  e^{-\left(r+\omega\right)\left(t_{d}-t\right)} \right|\mathcal{F}_{t}\right]  - \frac{b}{r\left(r+\omega\right)}
\end{array}
\end{equation}
Finally, the \textit{third addendum} is given by:
\begin{equation}\label{Finalthirdaddendum}
\displaystyle \int_{0}^{\infty} e^{-\omega\left(T-t\right)} \mathbb{E}\left[\left. \displaystyle \int_{t}^{T} b e^{-r\left(z-t\right)} \mathds{1}_{ \{ t_d >z \} } dz \right|\mathcal{F}_{t}\right] dT = \frac{b}{r\omega}  - \frac{b}{r(r+\omega)} -
\left(\frac{b}{r\omega} - \frac{b}{r(r+\omega)}\right) \mathbb{E}\left[\left.  e^{-\left(r+\omega\right)\left(t_{d}-t\right)} \right|\mathcal{F}_{t}\right]
\end{equation}
Collecting all the terms derived (\eqref{eq:laplacerecovery},\eqref{eq:laplacebond},\eqref{Finalthirdaddendum}) and recalling formula \eqref{eq:laplace}, we obtain the Laplace transform of the bond price as in \eqref{eq:bondprice}.
\end{proof}

\medskip

\begin{proof}[Proof of Proposition \ref{Prop::2}]
Denote by $\Pi$ the CDS spread which is payed continuously. Then,
by definition and based on the modeling framework, the Premium Leg is defined as:
\begin{equation}
\mbox{PremiumLeg}\left(x,t;T\right) = \mathbb{E}\left[\left. \displaystyle \Pi \int_{t}^{T} \ e^{-r\left(z-t\right)} \mathds{1}_{\left\{t_{d}\geq z\right\}}dz\right|\mathcal{F}_{t}\right]
\end{equation}
The Laplace transform of the Premium Leg, denoted by $F_{PremiumLeg}\left(x,t;\omega\right)$, can be computed as already shown in \eqref{Finalthirdaddendum}, obtaining:
\begin{equation}
 F_{PremiumLeg}\left(x,t;\omega\right) = \mathbb{E}\left[\left. \displaystyle  e^{-\left(r+\omega\right)\left(t_{d}-t\right)}\right|\mathcal{F}_{t}\right]\left(\frac{\Pi}{r\left(r+\omega\right)} - \frac{\Pi}{r\omega}\right) + \frac{\Pi}{r\omega} - \frac{\Pi}{r\left(r+\omega\right)} 
\end{equation}
Moreover, the Protection Leg, denoting with $1-\Upsilon$ the loss-given default, is given by:
\begin{equation} \label{eq:laplace_prem_leg}
\mbox{ProtectionLeg}(x,t;T) = \mathbb{E}\left[\left. \displaystyle e^{-r\left(t_{d}-t\right)}\left(1-\Upsilon\right)\mathds{1}_{\left\{T\geq t_{d} \geq t\right\}}\right|\mathcal{F}_{t}\right]  
\end{equation}
Similar as before, 
the Laplace transform of the Protection Leg, denoted by $F_{ProtectionLeg}\left(x,t;\omega\right)$, can be derived as in \eqref{eq:laplacerecovery}, obtaining:
\begin{equation} \label{eq:laplace_prot_leg}
F_{ProtectionLeg}\left(x,t;\omega\right) = \mathbb{E}\left[\left. e^{-\left(r+\omega\right)\left(t_{d}-t\right)}\right|\mathcal{F}_{t}\right]\frac{\left(1-\Upsilon\right)}{\omega} 
\end{equation}
Collecting the terms derived (\eqref{eq:laplace_prem_leg},  \eqref{eq:laplace_prot_leg}) and recalling
\eqref{eq:laplace}, the semi-analytical expression of the CDS spread in \eqref{eq:CDSsperadsSAS}, follows by choosing $\Pi$ that makes the Premium Leg equals to the Protection Leg.
\end{proof}

\section{The Credit-Risk Model Without Jump}\label{AppendixB}
In this appendix we summarize credit-risk structural model without jump and report the pricing formula for a CDS. 
In this simplified framework, the value of a firm is driven by a Geometric Brownian Motion, defined as:
\begin{equation}\label{eq:GBM}
\frac{dS_t}{S_{t}} = r dt + \sigma_s dW^Q_{t}
\end{equation}
where $S_t$ denotes the firm's asset value, $r$ is the risk-free interest rate and $\sigma_s$ is the volatility parameter and $W_t$ is a standard Weiner process under the risk neutral measure $Q$. 
As in the model defined in Section \ref{Sec::3}, the firm goes bankrupt when breaches from above a constant default barrier $S_{def}$, and it coincides with the first-time passage:
\begin{equation}
t_{d} = \inf \left\{ t > 0 : S_{t} \leq S_{def} \right\}, \quad t \in (0,T]
\end{equation}
Following \cite{brigo2013counterparty} and emplying the notation $z = \ln(S)$, the survival probability $\mathbb{P}^{surv}(z,t;T)$ computed in $t$ is defined as:
\begin{equation} \label{eq:survprob}
    \mathbb{P}^{surv}(z,t;T) = N \left( \frac{\ln{\frac{S_t}{S_{def}}}+m(T-t)} {\sigma \sqrt{T-t}} \right) - \left( \frac{S_{def}}{S_t} \right)^{\frac{2m}{\sigma^2}} N \left( \frac{\ln{\frac{S_{def}}{S_t}}+m(T-t)} {\sigma \sqrt{T-t}} \right) 
\end{equation} 
where $m = r - \frac{1}{2} \sigma_s^2$ and T indicates the maturity.

\subsection{Pricing a CDS in the No-Jump Framework} \label{Appendix:CDS_diffusion}
Having derived the survival probability $\mathbb{P}^{surv}(z,t;T)$ defined in \eqref{eq:survprob}, we follow the approach proposed in \cite{FangOosterleeCDS_closed} to derive the formula to price a CDS spread, denoted by $\Lambda$. It follows that:
\begin{equation}
    \Lambda(z,t;T) = \frac{(1-\Upsilon) \left(\displaystyle \int_t^T e^{-r(T-t)} d \mathbb{P}^{def} (z,t;s) \right) }{ \displaystyle \int_t^T e^{-r (s-t)} \mathbb{P}^{surv} (z,t;s)  ds }
\end{equation}
Integrating by parts and defining the default probability as $\mathbb{P}^{def} (z,t;T) = 1 - \mathbb{P}^{surv}(z,t;T)$, we obtain the final form of the pricing formula:
\begin{equation}\label{FormlaCDSsimple}
\Lambda(z,t;T) = (1-\Upsilon) \left( \frac{\displaystyle 1 - e^{-r(T-t)} \mathbb{P}^{surv}(z,t;T) }{ \displaystyle\int_t^T e^{-r (s-t)} \mathbb{P}^{surv}(z,t;s) ds } -r \right)
\end{equation}
The integral at the denominator of formula \eqref{FormlaCDSsimple} can be approximated using a classical quadrature scheme method, see, e.g., \cite{FangOosterleeCDS_closed}.

\section{Theoretical Term-Structure of the Credit Spread} \label{Appendix:theoreticalspread}
In the modeling framework outlined in Section \ref{Sec::3}, the parameters $\lambda$ and $\eta$ drive the transition risk whereby, the former models the arrival rate of more stringent green regulations, and the latter describes the magnitude of downward jumps on the firm's value. 
In fact, as underlined in \cite{AgliardiAgliardi2021}, the parameter $\eta$ describes the impact of stricter green policies on the value of a firm, in the manner whereby the higher is $\eta$, the greener is the firm. 

In this Appendix we conduct a sensitivity analysis on the parameter $\eta$, generating several credit spreads. First, in Figure \ref{fig:1}, we report the theoretical term-structure of a defaultable (green) bond and a CDS spread, generated by the model pricing formulas for different level of greenness. 
We observe an increment of the price of the bond as $\eta$ increases (greenness is higher). A lower bond price, corresponding to a lower $\eta$ value, guarantees a higher return to the investor, compensating for the higher exposure to transition risk.
\begin{figure}[H]
    \centering
    \includegraphics[scale=0.7]{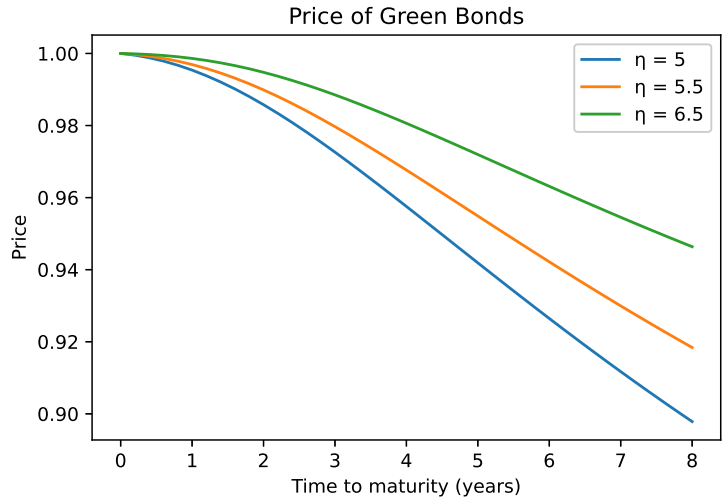}
    \caption{$V_t$ = 4,\ $V_{def}$ = 1,\ $\Upsilon = 0.6$,\ b = 0,\ r = 0.02 ,\ $\sigma = 0.2$,\ $\lambda = 0.4$. }
    \label{fig:1}
\end{figure} \noindent
Similarly, in Figure \ref{fig:2} we present the CDS spreads generated by the model. We can observe that, as $\eta$ increases, the CDS spread decreases. This indicates that the CDS spread of green companies is lower than that of brown companies.
\begin{figure}[H]
    \centering
    \includegraphics[scale=0.7]{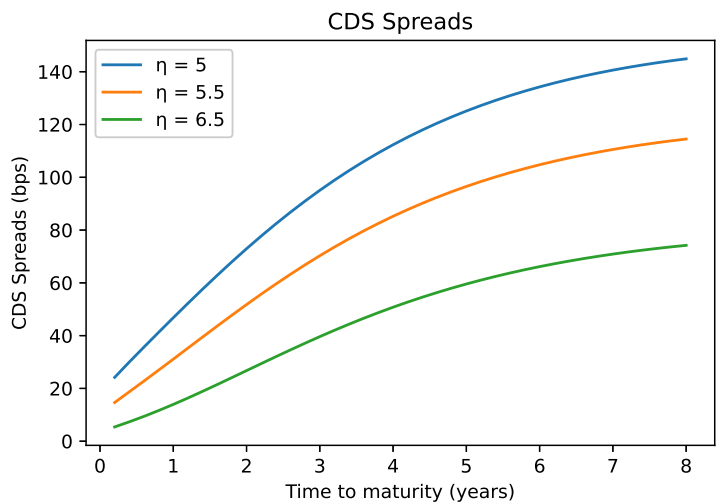}
    \caption{$V_t$ = 4,\ $V_{def}$ = 1,\ $\Upsilon$  = 0.6,\  r = 0.02 ,\ $\sigma = 0.2$,\ $\lambda = 0.4$. }
    \label{fig:2}
\end{figure} \noindent
In the last Figure \ref{fig:3} we derive the green spread generated by the model, computing the yields of a pure green bond and of a bond issued by a firm with a given shade of greenness. Specifically, a pure green bond is a bond issued by a firm that is not affected by transition risk, so that, its value follows a diffusion process without jumps. Assuming that the bond considered neither pays coupons nor refunds the holder with the recovery rate in the default event, the green spread is given by:
\begin{equation}
    \Psi(x,t;T,\eta) = \frac{- \log(  B(x,t;T,\eta) / ( e^{-r(T-t)} \mathbb{P}^{surv}({x,t;T} ) ) }{T-t}
\end{equation}
where $\mathbb{P}^{surv}({x,t;T} )$ is defined in \eqref{eq:survprob} and $B$ is the bond price derived by our model, under the same assumptions as the pure green bond. The theoretical term-structure of the green spreads for different shade of greenness is reported in Figure \ref{fig:3}.
\begin{figure}[H]
    \centering
    \includegraphics[scale=0.7]{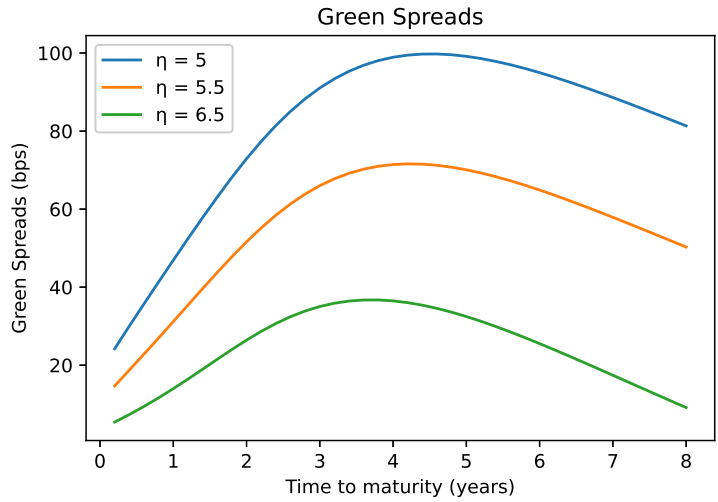}
    \caption{$V_t$ = 4,\ $V_{def}$ = 1,\ $\Upsilon$ = 0,\ b = 0, \ r = 0.02 ,\ $\sigma = 0.2$,\ $\lambda = 0.4$.}
    \label{fig:3}
\end{figure}

\newgeometry{top=0.2in, bottom=0.2in, left=1in, right=1in}

\section{Panel Quantile Regressions Tables} \label{Appendix:Results}
\begin{table}[H]
 \renewcommand{\arraystretch}{0.5}
 \begin{adjustwidth}{0cm}{}
 \scalebox{0.73}{%
\begin{tabular}{l|l|ccccccccc} \hline & \\[-1.5ex]
& & 0.1 &        0.2 &       0.3 &     0.4 &     0.5 &      0.6 &      0.7 &      0.8 &       0.9 \\ \hline  \hline & \\[-1.5ex] 
 \multirow{8}{*}{6 m} & \mbox{Return} &$-4.1068^{***}$ & $-2.3116^{***}$ & $-1.3085^{***}$ & $-0.8791^{***}$ &  $-0.6093^{***}$ & $-0.8423^{***}$ & $-1.4223^{***}$ &  $-2.4705^{***}$ &  $-5.1219^{***}$ \\  
 &  & (0.4655) &  (0.1055) &  (0.0400) &  (0.0203) &  (0.0125) &  (0.0164) &   (0.031) &   (0.078) &   (0.4157) \\  
 &  $\Delta$ Vol & $-3.3603^{***}$ & $-0.7522^{***}$ & $-0.2759^{***}$ & $-0.1066^{***}$ &  $-0.0086^{***}$ &  $0.0728^{***}$ &  $0.2226^{***}$ &  $0.5595^{***}$ &   $2.7067^{***}$ \\ 
 &  & (0.0262) &  (0.0061) &  (0.0024) &  (0.0013) &  (0.0008) &   (0.0010) &  (0.0019) &  (0.0044) &   (0.0221) \\  
 &  $\Delta$ MRI & $0.3477^{***}$ &   $0.3116^{***}$ &  $0.2888^{***}$ &  $0.2737^{***}$ &  $0.2700^{***}$ &   $0.2789^{***}$ &  $0.3211^{***}$ &  $0.3947^{***}$ &   $0.5663^{***}$ \\ 
 &  & (0.0178) &  (0.0030) &  (0.0010) &  (0.0004) &  (0.0002) &  (0.0003) &  (0.0006) &  (0.0017) &   (0.0115) \\  
 &  $\Delta$ TR & $0.0943^{***}$ &  $0.0760^{***}$ &  $0.0702^{***}$ &  $0.0663^{***}$ &  $0.0644^{***}$ &  $0.0717^{***}$ & $0.0815^{***}$ &  $0.0925^{***}$ &   $0.1421^{***}$ \\  
 & & (0.0134) &   (0.0028) &  (0.0010) &  (0.0004) &  (0.0002) &  (0.0003) &  (0.0007) &  (0.0019) &   (0.0105) \\  \hline  \\[-1.5ex]
 \multirow{8}{*}{1 y} & \mbox{Return} & $-6.3824^{***}$ & $-3.1756^{***}$ & $-1.9581^{***}$ & $-1.3818^{***}$ & $-0.9189^{***}$ & $-1.3005^{***}$ & $-2.1068^{***}$ & $-3.5178^{***}$ &  $-7.0436^{***}$ \\ 
 & & (0.5008) &  (0.1389) &  (0.0588) &  (0.0330) &  (0.0184) &  (0.0254) &  (0.0465) &  (0.1001) &   (0.4515) \\  
 & $\Delta$ Vol & $-3.8476^{***}$ & $-0.9763^{***}$ & $-0.4102^{***}$ & $-0.1749^{***}$ & $-0.0139^{***}$ &  $0.1111^{***}$ &  $0.3207^{***}$ &  $0.7686^{***}$ &   $2.9986^{***}$ \\ 
 & & (0.0278) &   (0.0081) &  (0.0036) &  (0.0021) &  (0.0011) &  (0.0016) &  (0.0028) &  (0.0055) &   (0.0231) \\   
 & $\Delta$ MRI & $0.4105^{***}$ &  $0.3998^{***}$ &  $0.3862^{***}$ &  $0.3745^{***}$ &  $0.3674^{***}$ &  $0.3767^{***}$ &  $0.4117^{***}$ &  $0.4782^{***}$ &   $0.6534^{***}$ \\  
 & & (0.0152) &  (0.0033) &  (0.0012) &  (0.0006) &  (0.0003) &  (0.0004) &  (0.0008) &   (0.0019) &   (0.0105) \\  
 & $\Delta$ TR & $0.0922^{***}$ &  $0.0923^{***}$ &  $0.0924^{***}$ &  $0.0897^{***}$ &   $0.0850^{***}$ &  $0.0937^{***}$ &   $0.1094^{***}$ &  $0.1210^{***}$ &   $0.1417^{***}$ \\ 
 & & (0.0126) &  (0.0032) &  (0.0012) &  (0.0006) &  (0.0003) &  (0.0004) &  (0.0008) &  (0.0020) &   (0.0101) \\  \hline & \\[-1.5ex]
 \multirow{8}{*}{2 y} & \mbox{Return} & $-7.6878^{***}$ & $-3.9118^{***}$ & $-2.7359^{***}$ & $-2.009^{***}$ & $-1.4553^{***}$ & $-2.0721^{***}$ & $-2.9864^{***}$ & $-4.7830^{***}$ &  $-9.2548^{***}$ \\
 & & (0.6422) &  (0.1765) &  (0.0978) &  (0.0598) &   (0.0374) &   (0.0509) &  (0.0810) &  (0.1555) &  (0.5562) \\   
  & $\Delta$ Vol & $-4.1547^{***}$ &  $-1.2675^{***}$ & $-0.6649^{***}$ & $-0.2695^{***}$ & $-0.0113^{***}$ &  $0.2052^{***}$ &  $0.5546^{***}$ &  $1.1956^{***}$ &   $3.3557^{***}$ \\
 & & (0.0340) &  (0.0101) &  (0.0059) &  (0.0037) &  (0.0023) &  (0.0031) &  (0.0047) &    (0.0085) &   (0.0272) \\    
 & $\Delta$ MRI & $0.4236^{***}$ &  $0.3769^{***}$ &  $0.3753^{***}$ &  $0.3645^{***}$ &  $0.3590^{***}$ &  $0.3673^{***}$ &  $0.4055^{***}$ &  $0.4567^{***}$ &   $0.5719^{***}$ \\ 
 & & (0.0134) &  (0.0029) &  (0.0014) &   (0.0008) &  (0.0004) &  (0.0006) &  (0.0010) &  (0.0022) &   (0.0093) \\
 & $\Delta$ TR & $0.1451^{***}$ &   $0.1486^{***}$ &  $0.1443^{***}$ &  $0.1409^{***}$ &  $0.1403^{***}$ &  $0.1521^{***}$ &  $0.1656^{***}$ &  $0.1790^{***}$ &   $0.1969^{***}$ \\
 & & (0.0135) &  (0.0033) &  (0.0016) &  0.000964 &  0.000583 &   0.00079 &  0.001298 &  0.002672 &   0.010565 \\ 
  \hline & \\[-1.5ex]
 \multirow{8}{*}{3 y} & \mbox{Return} & $-8.5255^{***}$ & $-4.6270^{***}$ & $-3.2277^{***}$ & $-2.5058^{***}$ & $-1.9138^{***}$ & $-2.6231^{***}$ & $-3.8275^{***}$ & $-5.9435^{***}$ & $-11.1866^{***}$ \\ 
 & & (0.7384) &  (0.2370) &  (0.1324) &  (0.0844) &  (0.0539) &  (0.0696) &  (0.1132) &  (0.2085) &   (0.6588) \\   
 & $\Delta$ Vol & $-4.3768^{***}$ & $-1.5770^{***}$ &  $-0.8194^{***}$ & $-0.3748^{***}$ & $-0.0381^{***}$ &  $0.2488^{***}$ &  $0.7247^{***}$ &  $1.4830^{***}$ &   $3.7591^{***}$ \\ 
 & & (0.0381) &  (0.0133) &  (0.0079) &  (0.0052) &  (0.0033) &  (0.0042) &  (0.0066) &  (0.0112) &   (0.0310) \\ 
 & $\Delta$ MRI & $0.3993^{***}$ &  $0.3560^{***}$ &  $0.3489^{***}$ &  $0.3461^{***}$ &  $0.3462^{***}$ &  $0.3526^{***}$ &  $0.3789^{***}$ &   $0.4192^{***}$ &   $0.5136^{***}$ \\ 
 & & (0.0110) &  (0.0029) &  (0.0014) &  (0.0008) &  (0.0005) &  (0.0006) &  (0.0011) &  (0.0022) &   (0.0079) \\   
 & $\Delta$ TR & $0.1607^{***}$ &  $0.1583^{***}$ &  $0.1554^{***}$ &   $0.1521^{***}$ &  $0.1493^{***}$ &  $0.1583^{***}$ &  $0.1666^{***}$ &  $0.1790^{***}$ &   $0.1986^{***}$ \\ 
 & &  (0.0126) &  (0.0036) &  (0.0018) &  (0.0011) &  (0.0006) &  (0.0008) &  (0.0014) &  (0.0028) &   (0.0097) \\ 
 \hline & \\[-1.5ex]
 \multirow{8}{*}{4 y} & \mbox{Return} & $-9.7164^{***}$ & $-5.0982^{***}$ & $-3.6035^{***}$ & $-2.7001^{***}$ & $-2.0107^{***}$ & $-2.7411^{***}$ & $-3.8935^{***}$ & $-6.4579^{***}$ & $-12.8842^{***}$ \\ 
 & & (0.8356) &  (0.2724) &  (0.1548) &  (0.1011) &  (0.0653) &  (0.0845) &  (0.1288) &  (0.2316) &   (0.7601) \\  
 &  $\Delta$ Vol &  $-4.6445^{***}$ & $-1.7249^{***}$ &  $-0.9787^{***}$ & $-0.4528^{***}$ & $-0.0575^{***}$ &  $0.2619^{***}$ &  $0.7880^{***}$ &  $1.6494^{***}$ &   $4.3665^{***}$ \\ 
 & & (0.0424) &  (0.0153) &  (0.0093) &  (0.0063) &  (0.0041) &  (0.0052) &  (0.0075) &  (0.0124) &   (0.0358) \\  
 & $\Delta$ MRI & $0.4402^{***}$ &  $0.4031^{***}$ &   $0.4047^{***}$ &  $0.4008^{***}$ &  $0.4016^{***}$ &   $0.4090^{***}$ &  $0.4380^{***}$ &  $0.4710^{***}$ &   $0.5695^{***}$ \\ 
 & & (0.0119) &  (0.0032) &  (0.0016) &  (0.0009) &  (0.0005) &  (0.0007) &  (0.0012) &  (0.0023) &   (0.0087) \\  
 & $\Delta$ TR & $0.1908^{***}$ &  $0.1864^{***}$ &  $0.1816^{***}$ &  $0.1788^{***}$ &  $0.1771^{***}$ &  $0.1873^{***}$ &  $0.2002^{***}$ &  $0.2213^{***}$ &   $0.2362^{***}$ \\  
 & & (0.0139) &  (0.0040) &  (0.0020) &  (0.0012) &  (0.0008) &  (0.0010) &   (0.0016) &  (0.0031) &   (0.0112) \\ 
 \hline & \\[-1.5ex]
 \multirow{8}{*}{5 y} & \mbox{Return} & $-11.4430^{***}$ & $-5.8459^{***}$ & $-4.2022^{***}$ & $-3.0990^{***}$ & $-2.4444^{***}$ & $-3.2634^{***}$ & $-4.5913^{***}$ & $-7.0986^{***}$ & $-14.0178^{***}$ \\ 
 & & (0.9349) &  (0.2954) &  (0.1726) &   (0.1101) &  (0.0720) &  (0.0908) &  (0.1408) &  (0.2307) &    (0.7192) \\  
 & $\Delta$ Vol & $-4.9478^{***}$ & $-1.9719^{***}$ & $-1.0746^{***}$ & $-0.4912^{***}$ & $-0.0566^{***}$ &  $0.2644^{***}$ &  $0.8293^{***}$ &  $1.7603^{***}$ &   $4.5607^{***}$ \\ 
 & & (0.0469) &  (0.0166) &  (0.0103) &   (0.0068) &  (0.0045) &  (0.0056) &  0.008272 &  (0.01237) &   (0.03361) \\   
 & $\Delta$ MRI & $0.4271^{***}$ &   $0.3940^{***}$ &  $0.4001^{***}$ &   $0.4037^{***}$ &   $0.4105^{***}$ &  $0.4188^{***}$ &  $0.4473^{***}$ &  $0.4818^{***}$ &   $0.5825^{***}$ \\ 
 & & (0.0129) &  (0.0033) &   (0.0017) &  (0.0010) &  (0.0006) &  (0.0007) &  (0.0012) &  (0.0022) &   (0.0080) \\  
 & $\Delta$ TR & $0.2106^{***}$ &  $0.2092^{***}$ &   $0.1978^{***}$ &  $0.1898^{***}$ &  $0.1839^{***}$ &  $0.1949^{***}$ &  $0.2127^{***}$ &  $0.2343^{***}$ &   $0.2505^{***}$ \\ 
 & & (0.0150) &  (0.0042) &  (0.0022) &  (0.0013) &  (0.0008) &  (0.0010) &  (0.0017) &  (0.0030) &   (0.0103) \\ 
 \hline & \\[-1.5ex]
 \multirow{8}{*}{7 y} & \mbox{Return} & $-13.7272^{***}$ & $-6.8163^{***}$ & $-4.6931^{***}$ & $-3.4492^{***}$ & $-2.7840^{***}$ & $-3.5714^{***}$ & $-4.9224^{***}$ & $-7.7603^{***}$ & $-15.7068^{***}$ \\ 
 & & (1.0001) &  (0.3354) &  (0.1851) &  (0.1176) &  (0.0831) &   (0.1011) &  (0.1497) &  (0.2427) &   (0.8547) \\  
 & $\Delta$ Vol & $-5.8357^{***}$ & $-2.2695^{***}$ & $-1.1539^{***}$ & $-0.5138^{***}$ & $-0.0590^{***}$ &  $0.3022^{***}$ &  $0.9254^{***}$ &  $2.0163^{***}$ &   $5.0416^{***}$ \\
 & & (0.0502) &  (0.0188) &  (0.0111) &  (0.0073) &  (0.0052) &  (0.0062) &  (0.0087) &  (0.0130) &   (0.0396) \\ 
 & $\Delta$ MRI & $0.3728^{***}$ &  $0.3363^{***}$ &  $0.3422^{***}$ &   $0.3419^{***}$ &  $0.3476^{***}$ &  $0.3536^{***}$ &  $0.3827^{***}$ &  $0.4242^{***}$ &   $0.5267^{***}$ \\ 
 & & (0.0129) &  0.003616 &  0.001758 &  0.001023 &  0.000689 &  0.000832 &  0.001282 &  0.002266 &   (0.0090) \\  
 & $\Delta$ TR & $0.2736^{***}$ &   $0.2676^{***}$ &  $0.2531^{***}$ &   $0.2533^{***}$ &   $0.2480^{***}$ &  $0.2623^{***}$ &  $0.2848^{***}$ &  $0.3086^{***}$ &   $0.3351^{***}$ \\
 & & (0.0157) &  (0.0047) &  (0.0023) &  (0.0014) &  (0.0009) &  (0.0011) &  (0.0018) &  (0.0031) &   (0.0121) \\   
  \hline & \\[-1.5ex]
 \multirow{8}{*}{10 y} & \mbox{Return} & $-14.6564^{***}$ & $-7.4115^{***}$ & $-5.1175^{***}$ & $-3.7870^{***}$ &  $-2.9852^{***}$ & $-3.5927^{***}$ &  $-5.1934^{***}$ & $-8.5166^{***}$ & $-16.3607^{***}$ \\ 
 & & (0.9858) &   (0.3367) &  (0.1842) &  (0.1188) &  (0.0819) &  (0.0943) &  (0.1479) &   (0.2567) &   (0.7713) \\  
 & $\Delta$ Vol & $-5.7388^{***}$ & $-2.3635^{***}$ & $-1.1795^{***}$ & $-0.5289^{***}$ & $-0.0846^{***}$ &  $0.2555^{***}$ &  $0.8601^{***}$ &   $2.1040^{***}$ &   $5.3533^{***}$ \\ 
 & & (0.0492) &  (0.0189) &  (0.0110) &  (0.0074) &  (0.0051) &  (0.0058) &  (0.0086) &  (0.0137) &   (0.0352) \\  
 & $\Delta$ MRI & $0.3569^{***}$ &  $0.3195^{***}$ &  $0.3240^{***}$ &  $0.3179^{***}$ &  $0.3184^{***}$ &  $0.3283^{***}$ &  $0.3515^{***}$ &  $0.3881^{***}$ &   $0.4769^{***}$ \\ 
 & & (0.0124) &  (0.0035) &  (0.0017) &  (0.0010) &  (0.0006) &  (0.0007) &  (0.0012) &  (0.0023) &   (0.0080) \\  
 & $\Delta$ TR & $0.258281^{***}$ &  $0.260511^{***}$ &  $0.239783^{***}$ &  $0.239518^{***}$ &  $0.238234^{***}$ &  $0.250074^{***}$ &  $0.273796^{***}$ &  $0.306784^{***}$ &   $0.349737^{***}$ \\ 
 & & (0.0152) &  (0.0046) &  (0.0023) &  (0.0014) &  (0.0009) &  (0.0010) &  (0.0017) &  (0.0032) &   (0.0107) \\   \hline & \\[-1.5ex]
 \multirow{8}{*}{20 y} & \mbox{Return} & $-14.7697^{***}$ & $-7.1142^{***}$ & $-4.6905^{***}$ & $-3.2396^{***}$ & $-2.6567^{***}$ &  $-3.0845^{***}$ & $-4.5915^{***}$ & $-7.8075^{***}$ & $-16.1003^{***}$ \\ 
 & & (0.9495) &   (0.2864) &  (0.1589) &   (0.1026) &  (0.0728) &  (0.0831) &  (0.1223) &  (0.2288) &  (0.7432) \\  
 & $\Delta$ Vol & $-5.7421^{***}$ &  $-2.2686^{***}$ & $-1.0401^{***}$ & $-0.4137^{***}$ &  $-0.0417^{***}$ &   $0.2859^{***}$ &  $0.8837^{***}$ &  $1.9942^{***}$ &   $5.1096^{***}$ \\ 
 & & (0.0478) &   (0.0162) &  (0.0096) &  (0.0064) &  (0.0046) &  (0.0051) &  (0.0072) &  (0.0121) &   (0.0345) \\   
 & $\Delta$ MRI & $0.2214^{***}$ &  $0.2061^{***}$ &  $0.2078^{***}$ &  $0.2045^{***}$ &  $0.2028^{***}$ &  $0.2092^{***}$ &  $0.2313^{***}$ &  $0.2630^{***}$ &   $0.3280^{***}$ \\ 
 & & (0.0107) & (0.0027) &  (0.0013) &  (0.0007) &  (0.0005) &  (0.0006) &  (0.0009) & (0.0019) &   (0.0074) \\ 
 & $\Delta$ TR & $0.2256^{***}$ &  $0.2127^{***}$ &  $0.1953^{***}$ &  $0.1874^{***}$ &  $0.1894^{***}$ &  $0.2012^{***}$ &  $0.2276^{***}$ &  $0.2568^{***}$ &   $0.2922^{***}$ \\ 
 & & (0.0135) &  (0.0037) &  (0.0018) &  (0.0011) &  (0.0007) &  (0.0009) &  (0.0013) &  (0.0027) &   (0.0100) \\  \hline & \\[-1.5ex]
 \multirow{8}{*}{30 y} & \mbox{Return} & $-14.5589^{***}$ & $-7.2858^{***}$ & $-4.8345^{***}$ & $-3.2554^{***}$ &  $-2.571^{***}$ & $-3.0110^{***}$ &  $-4.4499^{***}$ & $-7.6243^{***}$ & $-15.3198^{***}$ \\  
 & & (0.9252) &  (0.2788) &  (0.1450) &  (0.0936) &  (0.0655) &  (0.0735) &  (0.1109) &  (0.2160) &   (0.7042) \\    
 & $\Delta$ Vol & $-5.9629^{***}$ & $-2.2174^{***}$ & $-0.9481^{***}$ & $-0.3754^{***}$ & $-0.0258^{***}$ &  $0.2872^{***}$ &  $0.8230^{***}$ &  $1.9192^{***}$ &   $5.3181^{***}$ \\ 
 & & (0.0469) &  (0.0158) &  (0.0088) &  (0.0059) &  (0.0042) &  (0.0046) &   (0.0065) &  (0.0115) &   (0.0332) \\  
 & $\Delta$ MRI & $0.1936^{***}$ &  $0.1713^{***}$ &  $0.1679^{***}$ &  $0.1601^{***}$ &  $0.1596^{***}$ &  $0.1634^{***}$ &  $0.1785^{***}$ &  $0.2107^{***}$ &   $0.2743^{***}$ \\ 
 & & (0.0099) &  (0.0025) &   (0.0011) &  (0.0006) &  (0.0004) &  (0.0005) &  (0.0008) &  (0.0018) &   (0.0070) \\  
 & $\Delta$ TR & $0.1893^{***}$ &  $0.1715^{***}$ &  $0.1559^{***}$ &  $0.1541^{***}$ &   $0.1524^{***}$ &  $0.1661^{***}$ &  $0.1855^{***}$ &  $0.2130^{***}$ &   $0.2449^{***}$ \\  
 & & (0.0124) &  (0.0034) &  (0.0016) &  (0.0010) &  (0.0006) &  (0.0007) &  (0.0012) &  (0.0025) &   (0.0091) \\  \hline 
\end{tabular}}
\begin{tablenotes}
      \small
      \item $^{***} p < 0.001$
\end{tablenotes}
\end{adjustwidth}
    \caption{This table presents the results of the panel regression for the \textit{jump-diffusion} model, where the dependent variable is the first-difference model implied CDS spread ($\Delta \mathcal{\mbox{CDS}}_{i,t}^{model}$) computed between two consecutive days. The columns contains the estimates for the quantiles $\tau \in \{0.1,...,0.9\}$. Standard errors are reported between brackets ().}
    \label{tab:jump-diffusion-deltaCDSmodel}
\end{table}
\newpage  
\hspace{1in}
\begin{table}[H]
 \renewcommand{\arraystretch}{0.5}
 \begin{adjustwidth}{-0.5cm}{}
 \scalebox{0.73}{%
\begin{tabular}{l|l|ccccccccc} \hline & \\[-1.5ex]
& & 0.1 &        0.2 &       0.3 &     0.4 &     0.5 &      0.6 &      0.7 &      0.8 &       0.9 \\ \hline  \hline & \\[-1.5ex] 
 \multirow{6}{*}{6 m} & \mbox{Return} &  $-0.9614^{**}$ &  $-0.0350^{\diamond}$ & $-0.0001^{**}$ &  $-0.00002^{\diamond}$ &   $0.000002^{\diamond}$ &  $0.000003^{\diamond}$ &  $0.000004^{\diamond}$ &  $-0.00001^{\diamond}$ &    $0.0009^{\diamond}$ \\  
& & (0.2930) &   (0.0175) &  (0.00007) &  (0.00001) &   (0.000007) &  (0.000007) &  (0.000007) &   (0.0001) &   (0.0477) \\  
& \mbox{Vol} & $-2.4894^{***}$ &   $0.0539^{***}$ & $-0.00003^{***}$ & $-0.00002^{***}$ &  $-0.000005^{***}$ &  $0.000003^{***}$ &   $0.00002^{***}$ &   $0.00063^{***}$ &   $0.0611^{**}$ \\  
& & (0.0268) &   (0.0017) &  (0.000007) &  (0.000001) &   (0.000001) &  (0.000001) &  (0.000001) &   (0.00001) &   (0.0051) \\   
& \mbox{TR} & $-0.0377^{***}$ &  $-0.0103^{***}$ & $-0.00001^{***}$ & $-0.000001^{***}$ &       $0.0^{***}$ &      $0.0^{***}$ &      $0.0^{***}$ &  $-0.000004^{***}$ &  $-0.0003^{***}$ \\ 
& & (0.0007) &   (0.00005) &       (0.0) &       (0.0) &        (0.0) &       (0.0) &       (0.0) &        (0.0) &   (0.0001) \\    \hline & \\[-1.5ex]
\multirow{6}{*}{1 y} & \mbox{Return} & $-0.5000^{\diamond}$ &  $-0.0292^{\diamond}$ & $-0.366554^{***}$ & $-0.0002^{\diamond}$ &  $-0.0393^{\diamond}$ & $-1.0287^{***}$ & $-1.4334^{***}$ & $-1.7115^{***}$ &  $-2.3132^{*}$ \\  
& & (2.0729) &   (0.2365) &  (0.0797) &  (0.0464) &   (0.0426) &   (0.2276) &  (0.2666) &   (0.2746) &   (1.1636) \\  
& \mbox{Vol} & $-9.0090^{***}$ &   $-1.3912^{***}$ & $-0.2472^{***}$ & $-0.00005^{\diamond}$ &   $0.001489^{\diamond}$ &  $0.4965^{***}$ &  $1.5059^{***}$ &   $2.6191^{***}$ &   $5.774416^{***}$ \\  
& & (0.2043) &   (0.0230) &  (0.0079) &  (0.0046) &   (0.0042) &  (0.0230) &  (0.0270) &   (0.0277) &   (0.1163) \\    
& \mbox{TR} & $0.039776^{***}$ &    $0.00185^{**}$ &  $0.001208^{***}$ &      $0.0^{\diamond}$ &   $0.0001^{\diamond}$ &  $0.0014^{\diamond}$ & $-0.000176^{***}$ &   $0.007189^{***}$ &   $0.009777^{**}$ \\ 
& & (0.0053) &   (0.0005) &  (0.0002) &  (0.0001) &   (0.0001) &   (0.0005) &  (0.0006) &   (0.0007) &   (0.0029) \\    \hline & \\[-1.5ex]
 \multirow{6}{*}{2 y} & \mbox{Return} & $1.8717^{\diamond}$ &   $1.2692^{\diamond}$ &  $0.3357^{\diamond}$ & $-0.0186^{\diamond}$ &     $0.0374^{\diamond}$ &  $0.01779^{\diamond}$ & $-0.0514^{\diamond}$ &  $-0.3220^{\diamond}$ &   $-1.5636^{\diamond}$ \\ 
& & (2.5714) &   (0.7613) &  (0.2972) &  (0.1339) &   (0.1088) &  (0.1296) &  (0.4283) &   (0.3963) &   (0.8213) \\  
& \mbox{Vol} & $-13.6256^{***}$ &  $-4.5514^{***}$ & $-1.2938^{***}$ & $-0.3106^{***}$ &   $-0.1369^{***}$ & $-0.1881^{***}$ &  $1.00632^{***}$ &   $2.1511^{***}$ &   $6.9873^{***}$ \\ 
& & (0.3022) &   (0.0914) &  (0.0342) &  (0.0152) &   (0.0124) &  (0.0148) &  (0.0490) &   (0.0464) &   (0.1066) \\    
& \mbox{TR} & $0.0605^{***}$ &   $0.02805^{***}$ &  $0.01205^{***}$ &  $0.0069^{***}$ &   $0.0054^{***}$ &  $0.0098^{***}$ &  $0.0092^{***}$ &   $0.0093^{***}$ &  $-0.0161^{***}$ \\ 
& & (0.0056) &   (0.0016) &  (0.0006) &  (0.0002) &   (0.0002) &  (0.0002) &  (0.0008) &   0.0008) &   (0.0018)  \\   \hline & \\[-1.5ex]
 \multirow{6}{*}{3 y} & \mbox{Return} & $4.8017^{*}$ &   $4.4491^{***}$ &  $2.6799^{***}$ &  $2.7078^{***}$ &   $2.4222^{***}$ &  $1.3505^{***}$ &  $0.8565^{***}$ &   $1.3805^{***}$ &   $0.4613^{\diamond}$ \\  
& & (2.5142) &   (0.6963) &  (0.4859) &  (0.4142) &   (0.3829) &  (0.2543) &  (0.2263) &   (0.3366) &   (0.8128) \\ 
& \mbox{Vol} & $-9.6000^{***}$ &  $-3.9788^{***}$ &  $-2.1711^{***}$ & $-1.6477^{***}$ &  $-0.9967^{***}$ &  $0.5893^{***}$ &  $1.0509^{***}$ &   $2.0409^{***}$ &   $6.8891^{***}$ \\ 
& & (0.3248) &   (0.0874) &   (0.0604) &  (0.0519) &   (0.0483) &  (0.0324) &  (0.0289) &   (0.0440) &   (0.1105) \\  
& \mbox{TR} & $-0.004047$ &  $-0.0173^{***}$ & $-0.0184^{***}$ & $-0.0133^{***}$ &  $-0.0128^{***}$ & $-0.0203^{***}$ & $-0.0149^{***}$ &  $-0.0173^{***}$ &  $-0.0373^{***}$ \\ 
& & (0.0043) &   (0.0011) &  (0.0008) &  (0.0006) &    (0.0006) &  (0.0004) &  (0.0003) &   (0.0005) &   (0.0014) \\   \hline & \\[-1.5ex]
 \multirow{6}{*}{4 y} & \mbox{Return} & $0.7850^{\diamond}$ &   $0.497041^{*}$ &  $0.185274^{*}$ &  $0.002744^{*}$ &   $-0.00017^{\diamond}$ &   $0.00981^{\diamond}$ & $-0.548479^{\diamond}$ &  $-0.539229^{\diamond}$ &  $-1.559522^{\diamond}$ \\  
& & (0.7707) &   (0.2488) &   (0.0938) &  (0.0802) &   (0.0848) &  (0.1146) &  (0.2766) &   (0.3805) &     (0.7517) \\ 
& \mbox{Vol} & $-4.5027^{***}$ &  $-1.5483^{***}$ & $-0.6131^{***}$ & $-0.1837^{***}$ &  $-0.0926^{***}$ & $-0.0235^{*}$ &  $0.7201^{***}$ &   $2.2797^{***}$ &   $6.4748^{***}$ \\  
& & (0.1170) &   (0.0358) &   (0.0134) &  (0.0113) &   (0.0117) &  (0.0156) &  (0.0373) &   (0.0484) &   (0.0940) \\  
& \mbox{TR} & $0.01441^{***}$ &   $0.0064^{***}$ &  $0.0037^{***}$ &  $0.0013^{***}$ &   $0.0009^{***}$ &  $0.0026^{***}$ &  $0.0017^{***}$ &  $-0.0017^{***}$ &  $-0.0087^{***}$ \\  
& & (0.0011) &   (0.0003) &  (0.0001) &  (0.0001) &   (0.0001) &  (0.0001) &  (0.0003) &   (0.0004) &   (0.0009) \\   \hline & \\[-1.5ex]
 \multirow{6}{*}{5 y} & \mbox{Return} & $5.21413^{*}$ &   $0.975125^{*}$ & $-0.224915^{\diamond}$ & $-0.261743^{\diamond}$ &  $-0.657334^{**}$ &  $-0.78973^{*}$ & $-1.167283^{***}$ &  $-1.543307^{***}$ &  $-2.259644^{**}$ \\ 
& & (2.5557) &   (0.5411) &  (0.3899) &  (0.2167) &   (0.2479) &  (0.3236) &  (0.3037) &   (0.4103) &    (0.8594) \\   
& \mbox{Vol} & $-31.4475^{***}$ &  $-7.1585^{***}$ & $-1.7761^{***}$ & $-0.2479^{***}$ &   $0.0440^{\diamond}$ &   $0.7375^{***}$ &  $0.9706^{***}$ &   $3.2140^{***}$ &   $6.7653^{***}$ \\ 
& & (0.4325) &   (0.0827) &  (0.0573) &  (0.0316) &   (0.0356) &    (0.0461) &   (0.0414) &   (0.0529) &   (0.1084) \\   
& \mbox{TR} & $0.1200^{***}$ &   $0.0265^{***}$ &  $0.00714^{***}$ &  $0.00378^{***}$ &   $0.0050^{***}$ &  $0.0059^{***}$ &  $0.0126^{***}$ &   $0.0107^{***}$ &   $0.0189^{***}$ \\  
& & (0.0035) &   (0.0006) &  (0.0004) &  (0.0002) &   0.0002) &  (0.0003) &   (0.0003) &   (0.0004) &   (0.0008) \\   \hline & \\[-1.5ex]
 \multirow{6}{*}{7 y} & \mbox{Return} & $8.2547^{**}$ &   $0.7634^{\diamond}$ & $-2.0723^{*}$ &  $-0.9384^{\diamond}$ &  $-4.1131^{***}$ & $-4.6397^{***}$ & $-3.8456^{***}$ &  $-5.7417^{***}$ &  $-7.2098^{***}$ \\ 
& & (3.0663) &   (1.6973) &  (1.1137) &  (0.6278) &   0.8715) &  (0.9780) &  (1.0697) &   (1.3881) &   (1.6044) \\  
& \mbox{Vol} & $-58.1959^{***}$ & $-17.2456^{***}$ & $-5.5184^{***}$ &  $0.8484^{***}$ &   $3.3317^{***}$ &  $7.1436^{***}$ &   $8.1941^{***}$ &   $9.7285^{***}$ &  $14.4300^{***}$ \\ 
& & (0.5367) &   (0.2884) &  (0.1812) &  (0.0991) &   (0.1339) &   (0.1448) &  (0.1490) &   (0.1799) &   (0.2011) \\   
& \mbox{TR} & $0.1481^{***}$ &   $0.0473^{***}$ &  $0.0093^{***}$ & $-0.0028^{***}$ &   $0.0011^{\diamond}$ & $-0.0008^{\diamond}$ &  $0.0136^{***}$ &   $0.0295^{***}$ &    $0.0421^{***}$ \\  
& & (0.0031) &   (0.0017) &  (0.0010) &  (0.0005) &   (0.0008) &  (0.0008) &  (0.0009) &   (0.0011) &   (0.00123) \\   \hline & \\[-1.5ex]
 \multirow{6}{*}{10 y} & \mbox{Return} & $0.7107^{\diamond}$ &  $-2.3099^{\diamond}$ & $-4.1078^{*}$ & $-2.7753^{***}$ &  $-2.1724^{}$ & $-2.8258^{***}$ & $-2.7153^{*}$ &  $-4.5127^{**}$ &  $-1.7540^{\diamond}$ \\  
& & (5.3367) &   (2.8837) &  (2.3708) &  (1.1838) &   (0.6202) &  (0.8541) &  (1.2523) &   (1.6100) &    (2.3109) \\ 
& \mbox{Vol} & $-66.5803^{***}$ & $-13.4016^{***}$ & $-3.8290^{***}$ &  $2.3298^{***}$ &   $1.8154^{***}$ &  $6.5870^{***}$ &  $8.3139^{***}$ &  $11.7102^{***}$ &   $1.4856^{***}$ \\  
& & (0.7702) &   (0.4435) &  (0.3890) &  (0.1920) &   (0.0986) &  (0.1311) &  (0.1874) &   (0.2304) &   (0.3735) \\    
& \mbox{TR} & $0.0400^{***}$ &  $-0.0717^{***}$ &   $-0.0417^{***}$ & $-0.0244^{***}$ &  $-0.0054^{***}$ & $-0.0116^{***}$ &  $0.0076^{***}$ &   $0.0221^{***}$ &   $0.1431^{***}$ \\ 
& & (0.0040) &   (0.0023) &   (0.0020) &  (0.0010) &   (0.0005) &  (0.0006) &  (0.0009) &   (0.0012) &   (0.0019) \\    \hline & \\[-1.5ex]
\multirow{6}{*}{20 y} & \mbox{Return} & $28.8820^{**}$ &    $2.5063^{\diamond}$ &  $0.2324^{\diamond}$ & $-0.8500^{\diamond}$ &   $0.0291^{\diamond}$ &  $0.4854^{\diamond}$ &  $2.1068^{\diamond}$ &   $1.0833^{\diamond}$ &    $1.4380^{\diamond}$ \\ 
& & (9.1559) &   (4.4486) &  (2.8916) &  (1.7687) &   (1.6311) &  (1.4683) &  (1.3019) &   (1.3198) &   (1.3672) \\ 
& \mbox{Vol} & $-81.2161^{***}$ & $-17.2253^{***}$ &   $4.0578^{***}$ &  $8.2122^{***}$ &   $9.6627^{***}$ &  $8.2464^{***}$ &  $8.3767^{***}$ &  $10.3641^{***}$ &   $8.2334^{***}$ \\  
& & (1.2334) &   (0.6497) &  (0.4547) &  (0.2860) &   (0.2607) &  (0.2357) &  (0.2096) &   (0.2162) &    (0.2384) \\   
& \mbox{TR} & $0.04070^{***}$ &  $-0.0776^{***}$ & $-0.0810^{***}$ & $-0.0602^{***}$ &  $-0.0508^{***}$ & $-0.0363^{***}$ & $-0.0274^{***}$ &  $-0.0225^{***}$ &    $0.0099^{***}$ \\  
& & (0.0062) &   (0.0032) &  (0.0022) &  (0.0014) &   (0.0012) &  (0.0011) &  (0.0010) &   (0.0010) &   (0.0011) \\    \hline & \\[-1.5ex]
\multirow{6}{*}{30 y} & \mbox{Return} & $42.7162^{***}$ &  $14.7453^{**}$ &  $7.3467^{*}$ &  $1.7350^{\diamond}$ &  $-0.6425^{\diamond}$ & $-0.0578^{\diamond}$ &  $0.5922^{\diamond}$ &    $0.2169^{\diamond}$ &   $1.3834^{**}$ \\ 
& & (12.9178) &   (5.4523) &  (2.9647) &  (2.7206) &   (2.5582) &  (1.8428) &  (0.5233) &   (0.4528) &   (0.6551) \\  
& \mbox{Vol}& $-128.7750^{***}$ & $-31.0030^{***}$ & $-6.6050^{***}$ &  $9.0796^{***}$ &  $14.4791^{***}$ &  $9.5935^{***}$ &  $2.7888^{***}$ &   $0.3753^{***}$ &   $5.1133^{***}$ \\  
& & (1.7104) &   (0.7795) &  (0.4573) &  (0.4358) &  (0.4087) &  (0.2861) & (0.0795) &   (0.0744) &   (0.1183) \\  
& \mbox{TR} & $0.0032^{\diamond}$ &  $-0.1559^{***}$ & $-0.1590^{***}$ & $-0.1758^{***}$ &  $-0.1549^{***}$ & $-0.0823^{***}$ & $-0.0195^{***}$ &  $-0.0021^{***}$ &  $-0.0204^{***}$ \\  
& & (0.0084) &  (0.0037) &  (0.0021) &  (0.0020) &  (0.0019) &  (0.0013) &  (0.0003) &   (0.0003) &   (0.0005) \\  
\hline
\end{tabular}}
\begin{tablenotes}
      \small
      \item $^{***} p < 0.001$, $^{**} p < 0.01$, $^{*} p < 0.1$, $^{\diamond} \  0.1<p< 1 $
    \end{tablenotes}
    \end{adjustwidth}
    \caption{This table presents the results of the panel regression for the \textit{jump-diffusion} model, where the dependent variable is the error $\mathcal{E}_{i,t}$ between model and market CDS. The columns contains the estimates for the quantiles $\tau \in \{0.1,...,0.9\}$. Standard errors are reported between brackets ().}
    \label{tab:jump-diffusion-errorcalib}
\end{table}
\newpage
\hspace{1in}

\begin{table}[H]
 \renewcommand{\arraystretch}{0.5}
 \begin{adjustwidth}{-0.5cm}{}
 \scalebox{0.73}{
\begin{tabular}{l|l|ccccccccc} \hline & \\[-1.5ex]
& & 0.1 &        0.2 &       0.3 &     0.4 &     0.5 &      0.6 &      0.7 &      0.8 &       0.9 \\ \hline  \hline & \\[-1.5ex] 
\multirow{8}{*}{6 m} & \mbox{Return} &  $3.1055^{\diamond}$ &   $4.0041^{\diamond}$ &    $2.3562^{\diamond}$ &   $1.8887^{\diamond}$ &   $1.1614^{\diamond}$ &   $0.3790^{\diamond}$ &   $1.1465^{\diamond}$ &  $-0.0997^{\diamond}$ &   $0.0223^{\diamond}$ \\ 
& & (6.1933) &   (4.6106) &   (2.1093) &    (1.7640) &    (1.1208) &   (0.8851) &   (0.6865) &   (0.6457) &   (0.5893) \\  
& $\Delta$ Vol & $-252.2250^{***}$ & $-70.5118^{***}$ & $-43.6450^{***}$ & $-31.3166^{***}$ & $-21.7309^{***}$ & $-16.6341^{***}$ & $-12.8272^{***}$ &  $-8.6226^{***}$ &  $-4.9525^{***}$ \\ 
& & (0.3538) &   (0.2844) &   (0.1323) &   (0.1128) &   (0.0727) &   (0.0583) &   (0.0457) &   (0.0434) &   (0.0393) \\     
& $\Delta$ MRI & $-0.6949^{***}$ &  $-0.4702^{**}$ &  $-0.344013^{***}$ &  $-0.3260^{***}$ &  $-0.2174^{***}$ &  $-0.2156^{***}$ &   $-0.1851^{***}$ &   $-0.1981^{***}$ &  $-0.1701^{***}$ \\ 
& & (0.1973) &   (0.1444) &   (0.0566) &   (0.0418) &    (0.0235) &   (0.0165) &   (0.0112) &   (0.0089) &    (0.0071) \\ 
& $\Delta$ TR & $-0.1276^{\diamond}$ &  $-0.2353^{*}$ &  $-0.1569^{**}$ &  $-0.1882^{***}$ &  $-0.2093^{***}$ &   $-0.1997^{***}$ &  $-0.1477^{***}$ &  $-0.1259^{***}$ &  $-0.0801^{***}$ \\ 
& & (0.1416) &   (0.1075) &   (0.0484) &   (0.0405) &   (0.0256) &   (0.0198) &   (0.0149) &   (0.0134) &   (0.0116) \\   \hline & \\[-1.5ex]
\multirow{8}{*}{1 y} & Return & $8.9587^{\diamond}$ &   $5.5347^{\diamond}$ &   $4.2768^{\diamond}$ &   $3.9028^{*}$ &    $2.9270^{*}$ &    $1.5977^{\diamond}$ &   $0.3075^{\diamond}$ &  $-0.2208^{\diamond}$ &   $0.3661^{\diamond}$ \\  
& & (7.1282) &   (4.7106) &   (2.6958) &   (1.9723) &   (1.2941) &    (1.0615) &   (0.7826) &   (0.7758) &   (0.8562) \\  
& $\Delta$ Vol & $-205.6918^{***}$ & $-71.8713^{***}$ & $-46.1858^{***}$ &  $-32.9265^{***}$ & $-22.9435^{***}$ & $-16.9662^{***}$ & $-12.3693^{***}$ &  $-9.4609^{***}$ &  $-4.5862^{***}$ \\ 
& & (0.404351 &   (0.285805 &   (0.1685) &   (0.1256) &   (0.0839) &   (0.0701) &   (0.0526) &   (0.0528) &   (0.0589) \\  
& $\Delta$ MRI & $-0.322135^{*}$ &  $-0.218918^{*}$ &  $-0.2002^{***}$ &  $-0.2232^{***}$ &  $-0.1914^{***}$ &  $-0.1078^{***}$ &  $-0.0959^{***}$ &   $0.0168^{\diamond}$ &    $0.0630^{***}$ \\ 
& & (0.1783) &   (0.1148) &   (0.0575) &   (0.0378) &   (0.0225) &   (0.0173) &   (0.0125) &   (0.0134) &   (0.0182) \\  
& $\Delta$ TR & $-0.0744^{\diamond}$ &   $-0.1773^{*}$ &  $-0.16506^{***}$ &  $-0.2017^{***}$ &  $-0.1870^{***}$ &  $-0.2004^{***}$ &  $-0.1742^{***}$ &  $-0.1649^{***}$ &   $-0.0951^{***}$ \\ 
& & (0.1435) &   (0.0932) &   (0.0523) &   (0.0384) &   (0.0246) &   (0.0193) &   (0.0134) &   (0.0126) &   (0.0128) \\   \hline & \\[-1.5ex]
\multirow{8}{*}{2 y} & Return & $0.812715^{\diamond}$ &   $3.2147^{\diamond}$ &   $3.1538^{\diamond}$ &   $0.3426^{\diamond}$ &   $0.2168^{\diamond}$ &  $-0.1906^{\diamond}$ &  $-0.6844^{\diamond}$ &  $-1.0564^{\diamond}$ &   $-0.7422^{\diamond}$ \\  
& & (3.1085) &   (2.5786) &   (1.9360) &      (1.5240) &    (1.1930) &   (1.0468) &   (1.0290) &   (1.1760) &   (1.2994) \\ 
& $\Delta$ Vol & $-57.6960^{***}$ & $-41.7341^{***}$ & $-32.1719^{***}$ &  $-24.8106^{***}$ &  $-20.1381^{***}$ & $-15.6891^{***}$ & $-11.2799^{***}$ &  $-6.7300^{***}$ &   $0.6516^{***}$ \\  
& & (0.1838) &   (0.1609) &   (0.1218) &   (0.0962) &   (0.0759) &   (0.0667) &   (0.0654) &   (0.0745) &   (0.0791)\\   
& $\Delta$ MRI & $-0.0796^{\diamond}$ &  $-0.0295^{\diamond}$ &   $0.0227^{\diamond}$ &   $0.0476^{*}$ &   $0.0518^{***}$ &   $0.0406^{***}$ &   $0.0568^{***}$ &   $0.0658^{***}$ &   $0.1228^{***}$ \\ 
& & (0.0487) &   (0.0397) &   (0.0268) &   (0.01997) &   (0.0153) &   (0.0137) &   (0.0142) &   (0.0179) &   (0.0239) \\ 
& $\Delta$ TR & $-0.0644^{\diamond}$ &  $-0.0926^{*}$ &  $-0.1478^{***}$ &  $-0.2204^{***}$ &  $-0.1896^{***}$ &  $-0.1909^{***}$ &  $-0.1881^{***}$ &  $-0.1286^{***}$ &   $-0.0858^{***}$ \\ 
& & (0.0540) &   (0.0439) &   (0.0322) &    (0.0246) &   (0.0185) &   (0.0157) &   (0.0148) &   (0.0168) &   (0.0187)  \\   \hline & \\[-1.5ex]
\multirow{8}{*}{3 y} & Return & $-2.4544^{*}$ &  $-1.3447^{\diamond}$ &  $-0.9795^{\diamond}$ &  $-1.9624^{*}$ &  $-1.3671^{\diamond}$ &  $-1.2944^{\diamond}$ &  $-0.00002^{\diamond}$ &  $-0.000003^{\diamond}$ &  $-0.1707^{\diamond}$ \\ 
& & (1.3012) &    (1.2575) &   (0.9699) &   (0.9791) &   (1.0420) &   (1.0314) &  (0.4189) &   (0.3595) &   (1.8849) \\ 
& $\Delta$ Vol & $-26.1262^{***}$ & $-25.94676^{***}$ & $-23.5189^{***}$ & $-18.7001^{***}$ & $-15.4271^{***}$ &  $-8.240209^{\diamond}$ &  $-0.00001^{\diamond}$ &   $0.000001^{\diamond}$ &  $10.6677^{***}$ \\  
& & (0.0754) &   (0.0777) &   (0.0603) &  (0.0617) &   (0.0656) &   (0.0644) &   (0.0257) &   (0.0215) &   (0.1029) \\  
& $\Delta$ MRI & $-0.0317^{***}$ &   $0.0011^{\diamond}$ &   $0.0262^{***}$ &   $0.014433^{*}$ &   $0.0285^{***}$ &   $0.0241^{***}$ &        $0.0^{***}$ &        $0.0^{\diamond}$ &   $0.0015^{\diamond}$ \\ 
& & (0.0104) &   (0.0100) &    (0.0082) &   (0.0086) &   (0.0099) &   (0.0105) &    (0.0043) &   (0.0040) &   (0.0215) \\  
& $\Delta$ TR & $-0.173613^{***}$ &  $-0.1577^{***}$ &  $-0.1734^{***}$ &    $-0.1658^{***}$ &  $-0.1652^{***}$ &  $-0.1343^{***}$ &  $-0.000001^{\diamond}$ &       $0.0^{\diamond}$ &  $-0.008198^{\diamond}$ \\ 
& & (0.0211) &   (0.0176) &   (0.01336) &   (0.0130) &   (0.0132) &   (0.0126) &   (0.0048)&   (0.0042) &   (0.0219) \\   \hline & \\[-1.5ex]
\multirow{8}{*}{4 y} & Return & $-2.3510^{*}$ &  $-2.8125^{*}$ &  $-1.9876^{*}$ &  $-1.9116^{*}$ &  $-2.1913^{**}$ &   $-0.00005^{\diamond}$ &  $-4.1577^{***}$ &  $-4.0595^{*}$ &  $-3.9568^{*}$ \\  
& & (1.3822) &   (1.2432) &   (1.1208) &   (0.9088) &   (0.7883) &   (0.3988) &   (1.2386) &   (1.7229) &   (2.0509) \\  
& $\Delta$ Vol & $-17.8775^{***}$ & $-19.0492^{***}$ & $-14.2031^{***}$ &   $-9.311^{***}$ &  $-5.2140^{***}$ &  $-0.00001^{\diamond}$ &   $2.9852^{***}$ &  $14.8115^{***}$ &  $35.7578^{***}$ \\
& & (0.0860) &   (0.0780) &   (0.0708) &   (0.0575) &   (0.0495) &   (0.0247) &  (0.074861) &    (0.1001) &   (0.1060) \\ 
& $\Delta$ MRI & $0.0206^{*}$ &   $0.012609^{\diamond}$ &   $0.0477^{***}$ &   $0.0254^{***}$ &    $0.0415^{***}$ &   $0.000002^{\diamond}$ &   $0.0720^{***}$ &   $0.0459^{*}$ &   $0.0704^{**}$ \\ 
& & (0.0087) &   (0.0088) &   (0.0086) &   (0.0076) &   (0.0071) &   (0.0039) &   (0.0131) &     (0.0196) &   (0.0251) \\ 
& $\Delta$ TR & $-0.1175^{***}$ &  $-0.1096^{***}$ &  $ -0.1080^{***}$ &   $-0.1012^{***}$ &  $-0.0986^{***}$ &  $-0.000003^{\diamond}$ &  $-0.0791^{***}$ &  $-0.0345^{\diamond}$ &  $-0.02479^{\diamond}$ \\ 
& & (0.0192) &   (0.0165) &   (0.0145) &    (0.0115) &    (0.0097) &    (0.0048) &   (0.0149) &   (0.0209) &   (0.0262)   \\   \hline & \\[-1.5ex]
\multirow{8}{*}{5 y} & Return & $-0.0067^{\diamond}$ &  $-0.0301^{\diamond}$ &  $-0.000006^{\diamond}$ &  $-0.000003^{\diamond}$ &  $-0.0034^{\diamond}$ &  $-0.1061^{\diamond}$ &  $-0.1736^{\diamond}$ &   $-0.2643^{\diamond}$ &  $-0.3338^{\diamond}$ \\ 
& & (0.6872) &   (0.4328) &   (0.1026) &   (0.1153) &   (0.1387) &   (0.3506) &   (0.3659) &   (0.5416) &   (0.8066)\\ 
& $\Delta$ Vol & $-13.5941^{***}$ &  $-2.1323^{***}$ &  $-0.000008^{\diamond}$ &   $0.000001^{\diamond}$ &   $0.1432^{***}$ &   $5.0340^{***}$ &   $9.7155^{***}$ &     $19.419^{***}$ &  $33.1278^{***}$ \\ 
& & (0.0441) &   (0.0281) &   (0.0066) &   (0.0073) &   (0.0087) &   (0.0218) &   (0.0222) &   (0.0313) &   (0.0421)\\  
& $\Delta$ MRI & $0.0055^{\diamond}$ &   $0.0012^{\diamond}$ &       $0.0^{\diamond}$ &        $0.0^{\diamond}$ &   $0.0001^{\diamond}$ &   $0.0019^{\diamond}$ &   $0.0033^{\diamond}$ &    $0.0088^{\diamond}$ &    $0.0353^{\diamond}$ \\  
& & (0.0063) &   (0.0036) &   (0.0008) &   (0.000964 &   (0.0012) &   (0.0032) &   (0.0036) &   (0.0061) &    (0.0101)  \\
& $\Delta$ TR & $-0.005616^{\diamond}$ &  $-0.0010^{\diamond}$ &      $0.0^{\diamond}$ &       $0.0^{\diamond}$ &  $0.00007^{\diamond}$ &  $0.0012^{\diamond}$ &  $-0.0016^{\diamond}$ &  $-0.0038^{\diamond}$ &  $-0.0091^{\diamond}$ \\
& & (0.0086) &   (0.0054) &   (0.0012) &   (0.0014) &   (0.0016) &   (0.0041) &   (0.0043) &   (0.0068) &   (0.0107)  \\   \hline & \\[-1.5ex]
\multirow{8}{*}{7 y} & Return & $0.03120^{\diamond}$ &  $-0.1883^{\diamond}$ &   $-0.00001^{\diamond}$ &  $-1.5935^{\diamond}$ &  $-0.5418^{\diamond}$ &  $-0.5164^{\diamond}$ &  $-0.1722^{\diamond}$ &   $-0.7829^{\diamond}$ &  $-0.4409^{\diamond}$ \\ 
& & (0.7621) &   (0.2771) &   (0.3604) &   (0.7684) &   (0.8093) &  (0.7500) &   (0.7768)&   (0.8997) &   (1.1200) \\ 
& $\Delta$ Vol & $-6.9746^{***}$ &   $-0.4479^{***}$ &   $0.000001^{\diamond}$ &   $2.7797^{***}$ &     $6.6894^{***}$ &  $10.2932^{***}$ &  $13.7541^{***}$ &   $17.1105^{***}$ &  $23.5668^{***}$ \\ 
& & (0.0488) &   (0.0178) &  (0.0229) &   (0.0487) &   (0.0510) &   (0.0464) &   (0.0463) &   (0.0503) &    (0.0555) \\   
& $\Delta$ MRI & $0.0067^{\diamond}$ &   $0.0089^{***}$ &   $0.00001^{\diamond}$ &   $0.0294^{***}$ &   $0.0509^{***}$ &   $0.0475^{***}$ &   $0.0445^{***}$ &   $0.0409^{***}$ &   $0.03808^{***}$ \\ 
& & (0.0041) &    (0.0020) &   (0.0028) &   (0.0061) &   (0.0067) &   (0.0066) &    (0.0074) &   (0.0097) &   (0.0133) \\
& $\Delta$ TR & $0.0018^{\diamond}$ &   $0.0072^{*}$ &       $-0.0^{\diamond}$ &   $0.0434^{***}$ &   $0.0379^{***}$ &   $0.0446^{***}$ &    $0.0645^{***}$ &   $0.0463^{**}$ &   $0.02992^{**}$ \\  
& & (0.0073) &   (0.0031) &    (0.0042) &   (0.0090) &   (0.0095) &   (0.0089) &   (0.0096) &   (0.0116) &   (0.0157) \\    \hline & \\[-1.5ex]
\multirow{8}{*}{10 y} & Return & $2.1100^{\diamond}$ &   $0.9802^{\diamond}$ &   $0.00002^{\diamond}$ &    $0.00001^{\diamond}$ &    $0.6532^{\diamond}$ &    $1.0444^{\diamond}$ &   $1.0633^{\diamond}$ &   $1.3617^{\diamond}$ &   $0.5832^{\diamond}$ \\  
& & (1.4818) &   (0.9560) &   (0.2908) &   (0.3208) &   (0.5810) &   (0.9268) &   (0.8033) &   (1.0227) &   (1.2039) \\  
& $\Delta$ Vol & $-16.7535^{***}$ &  $-4.0113^{***}$ &  $-0.00001^{\diamond}$ &   $0.000001^{\diamond}$ &   $1.1344^{***}$ &   $6.2161^{***}$ &  $11.1312^{***}$ &  $14.2431^{***}$ &  $13.9336^{***}$ \\ 
& & (0.0860) &   (0.0593) &    (0.0181) &   (0.0203) &   (0.0367) &   (0.0588) &   (0.0509) &   (0.0644) &   (0.0737) \\  
& $\Delta$ MRI & $0.0025^{\diamond}$ &  $-0.00169^{\diamond}$ &     $0.0^{\diamond}$ &       $0.0^{\diamond}$ &  $0.0095^{\diamond}$ &  $0.0150^{*}$ &   $-0.0061^{\diamond}$ &   $0.00495^{\diamond}$ &   $0.02904^{***}$ \\  
& & (0.0118) &   (0.0076) &    (0.0023) &   (0.0026) &   (0.0047) &   (0.0072) &   (0.0057) &   (0.0071) &   (0.0087) \\ 
& $\Delta$ TR & $0.1310^{***}$ &   $0.0802^{***}$ &   $0.000002^{\diamond}$ &   $0.000001^{\diamond}$ &   $0.0545^{***}$ &    $0.0854^{***}$ &   $0.1030^{***}$ &   $0.1046^{***}$ &    $0.0753^{***}$ \\  
& & (0.0170) &   (0.0112) &    (0.0033) &   (0.0036) &   (0.0067) &   (0.0111) &   (0.0097) &   (0.0125) &    (0.0157)  \\   \hline & \\[-1.5ex]
\multirow{8}{*}{20 y} & Return & $8.6324^{**}$ &   $6.4267^{*}$ &   $5.2899^{*}$ &   $3.9391^{*}$ &   $4.2590^{\diamond}$ &   $3.0114^{\diamond}$ &   $0.8525^{\diamond}$ &   $0.9701^{\diamond}$ &   $0.6682^{\diamond}$ \\ 
& & (4.7369) &   (3.6232) &   (2.6405) &   (2.0422) &   (2.4208) &   (2.2034) &   (1.8653) &    (1.8132) &   (1.7537) \\ 
& $\Delta$ Vol & $-70.1374^{***}$ & $-42.8053^{***}$ & $-20.0820^{***}$ &  $-2.6321^{***}$ &   $7.5991^{***}$ &  $18.7052^{***}$ &  $24.8702^{***}$ &   $22.2789^{***}$ &   $18.8805^{***}$ \\ 
& & (0.2487) &    (0.2090) &   (0.1606) &   (0.1279) &   (0.1542) &   (0.1422) &   (0.1227) &   (0.1221) &   (0.1222) \\  
& $\Delta$ MRI & $-0.0836^{*}$ &  $-0.0975^{**}$ &  $-0.0860^{***}$ &  $-0.0967^{***}$ &  $-0.1124^{***}$ &  $-0.0780^{***}$ &  $-0.0797^{***}$ &  $-0.0750^{***}$ &  $-0.0811^{***}$ \\ 
& & (0.0430) &   (0.0324) &    (0.0217) &   (0.0161) &   (0.0181) &   (0.0153) &   (0.0123) &   (0.0117) &   (0.0097) \\  
& $\Delta$ TR & $0.0636^{\diamond}$ &   $0.1007^{*}$ &   $0.1095^{***}$ &   $0.1082^{***}$ &   $0.1238^{***}$ &   $0.1548^{***}$ &   $0.1754^{***}$ &   $0.1421^{***}$ &    $0.1095^{***}$ \\ 
& & (0.0554) &   (0.0396) &   (0.0283) &   (0.0221) &   (0.0263) &   (0.0245) &   (0.0211) &   (0.0206) &   (0.0184)  \\   \hline & \\[-1.5ex]
\multirow{8}{*}{30 y} & Return & $21.4858^{**}$ &   $7.7933^{*}$ &  $11.2974^{**}$ &   $8.8951^{**}$ &   $6.132^{**}$ &   $7.9725^{***}$ &   $3.1319^{\diamond}$ &   $0.0821^{\diamond}$ &   $-2.0407^{\diamond}$ \\ 
& & (7.3755) &   (4.5123) &     (3.8768) &   (2.5783) &   (2.1219) &   (2.3578) &   (2.4674) &   (1.9346) &   (1.7070) \\ 
& $\Delta$ Vol & $-98.2669^{***}$ &  $-64.1962^{***}$ & $-36.5270^{***}$ & $-16.8234^{***}$ &  $-2.6711^{***}$ &   $7.5035^{***}$ &  $17.3678^{***}$ &  $19.5688^{***}$ &  $21.5032^{***}$ \\
& & (0.3890) &   (0.2621) &   (0.2349) &   (0.1616) &   (0.1361) &   (0.1537) &   (0.1630) &   (0.1315) &    (0.1183) \\   
& $\Delta$ MRI & $-0.0493^{\diamond}$ &  $-0.1026^{**}$ &  $-0.1088^{***}$ &   $-0.1384^{***}$  &   $-0.1467^{***}$  &  $-0.1468^{***}$  &  $-0.1270^{***}$  &  $-0.1345^{***}$  &  $-0.1319^{***}$  \\ 
& & (0.0595) &   (0.0364) &   (0.0297) &   (0.0191) &   (0.0154) &   (0.0164) &   (0.0171) &   (0.0138) &   (0.0103) \\ 
& $\Delta$ TR & $0.0582^{\diamond}$ &    $0.0734^{\diamond}$ &    $0.1108^{**}$ &   $0.1449^{***}$ &   $0.1321^{***}$ &   $0.1523^{***}$ &    $0.1311^{***}$ &   $0.1001^{***}$ &   $0.1000^{***}$ \\  
& & (0.0808) &   (0.0473) &   (0.0403) &   (0.0267) &    (0.0222) &   (0.0253) &   (0.0271) &   (0.0212) &   (0.0165) \\ 
\hline
\end{tabular}
}
\begin{tablenotes}
      \small
      \item $^{***} p < 0.001$, $^{**} p < 0.01$, $^{*} p < 0.1$, $^{\diamond} \  0.1<p< 1 $
    \end{tablenotes}
    \end{adjustwidth}
    \caption{This table presents the results of the panel regression for the \textit{diffusion} model, where the dependent variable is first-difference model implied CDS spread ($\Delta \mathcal{\mbox{CDS}}_{i,t}^{model}$) computed between two consecutive days. The columns contains the estimates for the quantiles $\tau \in \{ 0.1,...,0.9 \}$. Standard errors are reported between brackets ().}
    \label{tab:diffusion-deltaCDSmodel}
\end{table}

\newpage
\hspace{1in}
\begin{table}[H]
 \renewcommand{\arraystretch}{0.5}
 \begin{adjustwidth}{-0.8cm}{}
 \scalebox{0.73}{
\begin{tabular}{l|l|ccccccccc} \hline & \\[-1.5ex]
& & 0.1 &        0.2 &       0.3 &     0.4 &     0.5 &      0.6 &      0.7 &      0.8 &       0.9 \\ \hline  \hline & \\[-1.5ex] 
\multirow{6}{*}{6 m} & \mbox{Return} & $-0.0054^{***}$ &  $-0.0016^{***}$ &   $-0.0009^{***}$ &    $-0.0008^{***}$ & $-0.0006^{***}$ &  $-0.0006^{***}$ &  $-0.0008^{***}$ &  $-0.0014^{***}$ &    $-0.0072^{***}$ \\ 
& & (0.0004) &   (0.00008) &   (0.00004) &   (0.00003) &  (0.00002) &   (0.00002) &   (0.00003) &    (0.00007) &   (0.0006) \\  
& Vol & $-0.0035^{***}$ &  $-0.0015^{***}$ &   $-0.0009^{***}$ &  $-0.0006^{***}$ & $-0.0002^{***}$ &  $-0.0001^{***}$ &  $-0.00008^{***}$ &  $-0.00002^{***}$ &  $0.0014^{***}$ \\  
& & (0.000038) &   (0.000008) &   (0.000004) &   (0.000003) &  (0.000003) &   (0.000002) &   (0.000003) &   (0.000007) &   (0.000058) \\ 
& TR & $-0.000004^{***}$ &  $-0.000007^{***}$ &  $-0.000006^{***}$ &  $-0.000002^{***}$ &  $0.000001^{***}$ &   $0.000004^{***}$ &   $0.000012^{***}$ &    $0.00002^{***}$ &   $0.000031^{***}$ \\ 
& & (0.0) &        (0.0) &        (0.0) &        (0.0) &       (0.0) &        (0.0) &        (0.0) &        (0.0) &   (0.0)  \\   \hline & \\[-1.5ex]
\multirow{6}{*}{1 y} &  \mbox{Return} & $-1.6624^{***}$ &  $-0.2780^{***}$ &  $-0.0817^{***}$ &  $-0.0295^{***}$ & $-0.0142^{***}$ &  $-0.0127^{***}$ &   $-0.0247^{***}$ &  $-0.1075^{***}$ &  $-1.0267^{***}$ \\ 
& & (0.2130) &   (0.0135) &   (0.0017) &    (0.0003) &  (0.0001) &   (0.0001) &   (0.0002) &   (0.0019) &   (0.0617) \\   
& Vol & $-0.5075^{***}$ &  $-0.0599^{***}$ &  $-0.0124^{***}$ &  $-0.0035^{***}$ & $-0.0012^{***}$ &  $-0.0004^{***}$ &   $0.0004^{***}$ &   $0.0070^{***}$ &    $0.13367^{***}$ \\ 
& & (0.0192) &   (0.0013) &   (0.0001) &   (0.00003) &  (0.00001) &   (0.00001) &  (0.00002) &   (0.0001) &   (0.0050) \\    
& TR & $0.0012^{*}$ &   $0.0003^{***}$ &   $0.00009^{***}$ &   $0.00003^{***}$ &  $0.00001^{***}$ &   $0.00001^{***}$ &   $0.00001^{***}$ &   $0.00004^{***}$ &   $0.00061^{***}$ \\  
& & (0.0005) &   (0.00003) &   (0.000005) &   (0.000001) &       (0.0) &        (0.0) &   (0.000001) &   (0.000005) &   (0.0001) \\   \hline & \\[-1.5ex]
\multirow{6}{*}{2 y} & \mbox{Return}  & $-13.8329^{***}$ &  $-4.9175^{***}$ &  $-2.5092^{***}$ &  $-1.4395^{***}$ & $-0.6667^{***}$ &  $-0.7363^{***}$ &  $-1.8540^{***}$ &  $-4.6526^{***}$ & $-15.0207^{***}$ \\ 
& & (0.8864) &    (0.1217) &   (0.0396) &   (0.01622) &  (0.0059) &   (0.0062) &   (0.0197) &   (0.0806) &    (0.5781) \\ 
& Vol & $-3.3682^{***}$ &  $-0.8664^{***}$ &  $-0.3742^{***}$ &  $-0.1533^{***}$ & $-0.0378^{***}$ &  $-0.0106^{***}$ &   $0.0750^{***}$&   $0.3574^{***}$ &   $1.6894^{***}$ \\ 
& & (0.0935) &   (0.0137) &   (0.0045) &   (0.0018) &  (0.0006) &   (0.0006) &   (0.0020) &   (0.0081) &   (0.0544) \\   
& TR & $0.0042^{*}$ &   $0.0027^{***}$ &   $0.0018^{***}$ &   $0.0009^{***}$ &  $0.0003^{***}$ &   $0.0003^{***}$ &    $0.0004^{***}$ &   $0.0009^{***}$ &   $0.0058^{***}$ \\  
& & (0.0017) &   (0.0002) &   (0.00008) &   (0.00003) &  (0.00001) &   (0.00001) &  (0.00003) &   (0.0001) &   (0.0010) \\   \hline & \\[-1.5ex]
\multirow{6}{*}{3 y} & \mbox{Return} & $-26.2692^{***}$  & $-11.8761^{***}$  &  $-7.0762^{***}$  &  $-4.2664^{***}$  & $-1.7042^{***}$  &  $-2.2444^{***}$  &  $-6.3563^{***}$  &  $-13.2778^{***}$  & $-33.4054^{***}$  \\
& & (1.2244)5 &   (0.242857 &   (0.1011) &   (0.0464) &  (0.0155) &   (0.0193) &   (0.0707) &   (0.2090) &    (1.0716) \\    
& Vol & $-6.4163^{***}$  &  $-1.9866^{***}$  &  $-0.9455^{***}$  &  $-0.4107^{***}$  &   $-0.0745^{\diamond}$  &   $0.0025^{***}$  &   $0.2890^{***}$  &   $0.9586^{***}$  &   $3.7684^{***}$  \\
& & (0.1444) &   (0.0300) &   (0.0126) &   (0.0058) &   (0.0019) &   (0.0023) &   (0.0082) &   (0.0231) &    (0.1127) \\ 
& TR & $0.0059^{**}$  &   $0.00417^{***}$  &   $0.0025^{***}$  &   $0.0015^{***}$  &  $0.0003^{***}$  &    $0.0006^{***}$  &   $0.0016^{***}$  &   $0.00424^{***}$  &   $0.0109^{***}$  \\ 
& & (0.0019) &   (0.0004) &    (0.0001) &   (0.00007) &  (0.00002) &   (0.00003) &   (0.0001) &   (0.0003) &   (0.0015)  \\   \hline & \\[-1.5ex]
\multirow{6}{*}{4 y} & \mbox{Return} & $-36.10705^{***}$ &  $-18.69365^{***}$ & $-11.91755^{***}$ &  $-6.90235^{***}$ & $-2.20075^{***}$ &  $-4.97775^{***}$ & $-12.314395^{***}$ & $-23.40855^{***}$ & $-52.07375^{***}$ \\ 
& & (1.4033) &   (0.3598) &   (0.1634) &   (0.0728) &  (0.0212) &   (0.0440) &   (0.1442) &   (0.3618) &   (1.6890) \\
& Vol & $-9.51225^{***}$ &  $-3.51745^{***}$ &  $-1.60055^{***}$ &  $-0.60315^{***}$ & $-0.05025^{***}$ &   $0.2068^{***}$ &   $1.0076^{***}$ &   $2.40285^{***}$ &   $7.50855^{***}$ \\
& & (0.1708) &   (0.0474) &   (0.0220) &   (0.0100) &  (0.0029) &   (0.0059) &   (0.0185) &   (0.0435) &    (0.1905) \\ 
& TR & $-0.000045^{\diamond}$ &    $0.0027^{***}$ &   $0.0012^{***}$&   $0.0011^{***}$ &  $0.0001^{\diamond}$ &   $0.00007^{***}$ &   $0.0010^{***}$ &   $0.0054^{***}$&   $0.0174^{***}$ \\  
& & (0.0017) &   (0.0004) &   (0.0002) &   (0.0001) &   (0.00003) &   (0.00006) &   (0.00018) &   (0.0004) &   (0.0019)  \\   \hline & \\[-1.5ex]
\multirow{6}{*}{5 y} & \mbox{Return} & $-42.2594^{***}$ & $-23.1555^{***}$ & $-14.8026^{***}$ &  $-8.5493^{***}$ & $-2.6291^{***}$ &  $-6.8665^{***}$ & $-16.6834^{***}$ & $-30.2765^{***}$ & $-61.5915^{***}$ \\
& & (1.5278) &   (0.4335) &   (0.2051) &   (0.0908) &  (0.02554) &   (0.06133) &   (0.1973) &    (0.4869) &   (2.0636) \\  
& Vol & $-11.4517^{***}$ &  $-4.5529^{***}$ &  $-2.0271^{***}$ &  $-0.7699^{***}$ & $-0.0333^{***}$ &   $0.4234^{***}$ &   $1.6928^{***}$ &   $3.6042^{***}$ &   $9.6731^{***}$ \\ 
& & (0.1816) &   (0.0573) &   (0.0283) &   (0.0129) &  (0.0036) &    (0.0085) &    (0.0260) &   (0.0589) &   (0.2279) \\   
& TR & $-0.0104^{***}$ &  $-0.0016^{***}$ &  $-0.0009^{***}$ &   $0.0011^{\diamond}$ &  $0.00002^{***}$ &  $-0.0005^{*}$ &    $0.00036^{***}$ &   $0.0071^{***}$ &   $0.0268^{***}$ \\  
& & (0.0014) &   (0.0004) &    (0.0002) &   (0.0001) &   (0.00003) &    (0.00007) &   (0.0002) &   (0.0004) &   (0.0018)  \\   \hline & \\[-1.5ex]
\multirow{6}{*}{7 y} &  \mbox{Return}  & $-46.8499^{***}$ & $-26.1987^{***}$ & $-17.4813^{***}$ & $-10.2383^{***}$ &  $-3.8535^{***}$ &  $-9.4837^{***}$ & $-20.5886^{***}$ & $-36.3619^{***}$ & $-70.0557^{***}$ \\ 
& & (1.7590) &   (0.4692) &   (0.2295) &   (0.1094) &  (0.0353) &    (0.0865) &    (0.2504) &   (0.6049) &   (2.4837) \\  
& Vol & $-15.2753^{***}$ &  $-5.8302^{***}$ &  $-2.5154^{***}$ &  $-0.9282^{***}$ & $-0.0007^{\diamond}$ &   $0.8064^{***}$ &   $2.5254^{***}$ &   $5.7199^{***}$ &  $13.2415^{***}$ \\ 
& & (0.2109) &   (0.0648) &   (0.0334) &   (0.0165) &  (0.0054) &   (0.0128) &   (0.0348) &   (0.0762) &   (0.2800) \\ 
& TR & $-0.0135^{***}$ &  $-0.0050^{***}$ &  $-0.0029^{***}$ &   $0.0006^{***}$ & $-0.000004^{\diamond}$ &  $-0.0010^{***}$ &   $0.0004^{*}$ &   $0.00613^{***}$ &   $0.0344^{***}$ \\
& & (0.0012) & 	(0.0003) & (0.0002) & 	(0.00009) & (0.00003) & (0.00009) &  	(0.0002) &	(0.0004) &	(0.0016) \\   \hline & \\[-1.5ex]
\multirow{6}{*}{10 y} & \mbox{Return} & $-43.8408^{***}$ & $-25.1945^{***}$ & $-16.6703^{***}$ &  $-9.3483^{***}$ & $-5.2453^{***}$ & $-10.2858^{***}$ & $-20.1126^{***}$ & $-35.7491^{***}$ & $-67.8181^{***}$ \\
& & (1.7411) &   (0.4505) &   (0.2193) &   (0.0950) &  (0.0478) &   (0.0951) &   (0.2432) &   (0.6098) &   (2.3496) \\  
& Vol & $-17.2013^{***}$ &  $-5.5618^{***}$ &  $-2.1117^{***}$ &  $-0.7263^{***}$ &   $0.0654^{***}$ &   $0.9383^{***}$ &   $2.4969^{***}$ &   $6.2195^{***}$ &  $14.2604^{***}$ \\ 
& & (0.2135) &   (0.0642) &   (0.0332) &   (0.0148) &  (0.0076) &   (0.0146) &   (0.0350) &   (0.0793) &   (0.2708) \\ 
& TR & $-0.0053^{***}$ &  $-0.0053^{***}$ &  $-0.0036^{***}$ &   $0.00001^{\diamond}$ & $-0.00007^{*}$ &  $-0.0007^{***}$ &   $0.0013^{***}$ &   $0.0048^{***}$ &   $0.0314^{***}$ \\ 
& & (0.0011) &   (0.0003) &   (0.0001) &   (0.00007) &   (0.00004) &   (0.00007) &   (0.0001) &   (0.0004) &   (0.0014) \\   \hline & \\[-1.5ex]
\multirow{6}{*}{20 y} & \mbox{Return} & $-37.1418^{***}$ & $-21.2232^{***}$ & $-13.2271^{***}$ &  $-8.0651^{***}$ & $-6.0049^{***}$ &  $-9.4217^{***}$ & $-16.3020^{***}$ & $-30.2802^{***}$ & $-57.6025^{***}$ \\  
& & (1.5544) &   (0.4114) &    (0.1789) &   (0.0862) &  (0.0575) &    (0.0942) &   (0.1954) &   (0.5471) &   (2.2594) \\ 
& Vol & $-18.1273^{***}$ &  $-5.7746^{***}$ &  $-1.8715^{***}$ &  $-0.5699^{***}$ &  $0.1209^{***}$ &    $0.8852^{***}$ &  $2.4357^{***}$ &   $6.9969^{***}$ &  $18.9665^{***}$ \\ 
& & (0.1916) &   (0.0599) &   (0.0274) &   (0.0135) &  (0.0091) &   (0.0146) &   (0.0284) &   0.(0727) &   (0.2568) \\ 
& TR & $0.0024^{*}$ &  $-0.0022^{***}$ &  $-0.0029^{***}$ &  $-0.0007^{***}$ &  $0.00002^{\diamond}$ &   $0.0005^{***}$ &   $0.0018^{***}$ &    $0.0003^{***}$ &   $0.0064^{***}$ \\ 
& & (0.0009) &   (0.0002) &   (0.0001) &   (0.00006) &  (0.00004) &   (0.00007) &   (0.0001) &   (0.0003) &   (0.0012) \\   \hline & \\[-1.5ex]
\multirow{6}{*}{30 y} & \mbox{Return} & $-32.9638^{***}$ & $-18.0532^{***}$ & $-11.3636^{***}$ &  $-7.0519^{***}$ & $-5.4903^{***}$ &  $-8.1607^{***}$ & $-13.4919^{***}$ & $-25.3438^{***}$ & $-51.6025^{***}$ \\
& & (1.4494) &   (0.3407) &   (0.1573) &   (0.0767) &  (0.0562) &   (0.0861) &   (0.1649) &   (0.4794) &   (1.9147) \\  
 & Vol & $-19.0081^{***}$ &  $-5.2461^{***}$ &  $-1.8281^{***}$ &  $-0.4864^{***}$ &  $0.0826^{***}$ &   $0.6800^{***}$ &   $1.8968^{***}$ &   $6.1094^{***}$ &  $19.3046^{***}$ \\
& & (0.1786) &    (0.04976) &    (0.02432) &   (0.0121) &  (0.0089) &   (0.0133) &   (0.0240) &   (0.0638) &   (0.2241) \\  
& TR & $0.0115^{***}$ &  $-0.0013^{***}$ &  $-0.0021^{***}$ &  $-0.0009^{***}$ &  $0.0001^{***}$ &   $0.0013^{***}$ &   $0.0027^{***}$ &   $0.0005^{*}$ &  $-0.0030^{***}$ \\  
& & (0.0008) &   (0.0002) &   (0.0001) &   (0.00005) &  (0.00004) &   (0.00006) &   (0.0001) &   (0.0003) &   (0.0010) \\ 
\hline
\end{tabular}
}
\begin{tablenotes}
\item $^{***} p < 0.001$, $^{**} p < 0.01$, $^{*} p < 0.1$, $^{\diamond} \  0.1<p< 1 $
\end{tablenotes}
    \end{adjustwidth}
    \caption{This table presents the results of the panel regression for the \textit{diffusion} model, where the dependent variable is the error $\mathcal{E}_{i,t}$ between model and market CDS. The columns contains the estimates for the quantiles $\tau \in \{ 0.1,...,0.9 \} $. Standard errors are reported between brackets ().}
    \label{tab:diffusion-errorcalib}
\end{table}

\restoregeometry

\newpage
\section{Descriptive Statistics} \label{Appendix:statistics}
This appendix reports the descriptive statistics for the variables we employed in the analysis.
\\
\begin{table}[H]
\centering
 \scalebox{0.8}{
\begin{tabular}{llrrrrrrr}
\hline
\hline
Variable             & observations &     mean &   median &     std &        skew &  kurtosis \\
\hline
 CDS market 6m       &       105756 &  21.9188 &   8.3700 & 73.4032 &   12.0192 &  180.7346 \\ 
 CDS market 1y       &       105756 &  26.8150 &  11.1900 & 81.4541 &   11.4796 &  171.6312 \\ 
 CDS market 2y       &       105756 &  38.4486 &  20.1100 & 87.1869 &     9.4004 &  115.8410 \\ 
 CDS market 3y       &       105756 &  50.1026 &  29.6000 & 91.0793 &    7.9210 &   83.6267 \\
 CDS market 4y       &       105756 &  64.0296 &  42.3700 & 92.5535 &   6.7104 &   60.9326 \\
 CDS market 5y       &       105756 &  77.8766 &  55.3400 & 95.4347 &    5.8134 &   46.9894 \\ 
 CDS market 7y       &       105756 &  99.9479 &  76.6000 & 96.8531 &   4.8107 &   33.5819 \\
 CDS market 10y      &       105756 & 115.5844 &  91.5950 & 96.3777 &   4.2664 &   27.5051 \\
 CDS market 20y      &       105756 & 125.1668 & 102.4250 & 95.7793 &   3.8413 &   23.0704 \\
 CDS market 30y      &       105756 & 130.6239 & 107.5200 & 95.8474 &   3.6567 &   21.3859 \\        \hline     
 Ret                 &       105756 &   0.0001 &   0.0003 &  0.0186 &   -0.5487 &   20.8007 \\     
 Vol                 &       105756 &   0.2682 &   0.2365 &  0.1165 &   1.9622 &    6.3121 \\  
 Carbon price        &         1259 &  24.9186 &  23.4000 & 17.0163 &   1.1675 &    1.1326 \\  
 TR 6m &                       1259 &   6.9138 &   4.2444 &  6.9662 &   3.0766 &   12.6582 \\      
 TR 1y &                       1259 &   8.4419 &   5.7428 &  8.0094 &   3.1992 &   13.6579 \\     
 TR 2y &                       1259 &  13.2221 &  11.4281 &  9.7826 &   2.8476 &   11.2722 \\     
 TR 3y &                       1259 &  18.9923 &  18.2223 & 11.2428 &   2.6586 &   10.5240 \\       
 TR 4y &                       1259 &  26.5263 &  25.3949 & 11.3736 &   2.5650 &   10.3813 \\      
 TR 5y &                       1259 &  34.1395 &  33.0266 & 11.6649 &   2.4313 &    9.5039 \\       
 TR 7y &                       1259 &  47.3614 &  46.2781 & 11.3790 &   2.2279 &    8.4454 \\    
 TR 10y &                      1259 &  54.7323 &  54.2707 & 11.1107 &   2.0062 &    7.6375 \\   
 TR 20y &                      1259 &  58.0126 &  57.5406 & 10.9802 &   1.9966 &    7.6018 \\     
 TR 30y &                      1259 &  60.0840 &  59.6963 & 11.0472 &   1.7956 &    6.9918 \\     
 MRI 6m &                    105756 &  11.2854 &   7.7400 & 12.6435 &    4.5012 &   22.3518 \\      
 MRI 1y &                    105756 &  14.5828 &  10.4800 & 15.3598 &    4.4879 &   23.0778 \\      
 MRI 2y &                    105756 &  24.2714 &  18.9400 & 21.2326 &    3.6527 &   16.7962 \\      
 MRI 3y &                    105756 &  34.5308 &  28.3100 & 27.1528 &    3.0949 &   13.5697 \\     
 MRI 4y &                    105756 &  47.9410 &  42.0400 & 33.1248 &    2.5021 &    8.6156 \\     
 MRI 5y &                    105756 &  61.4777 &  54.5300 & 39.2497 &   2.1292 &    5.9447 \\   
 MRI 7y &                    105756 &  84.6142 &  72.9750 & 46.4385 &   1.7703 &    3.8712 \\    
 MRI 10y &                   105756 & 101.1261 &  87.6500 & 51.0265 &   1.5869 &    3.0508 \\   
 MRI 20y &                   105756 & 109.7678 &  95.9250 & 53.6013 &   1.5815 &    3.0196 \\   
 MRI 30y &                   105756 & 114.9311 &  98.7400 & 54.3299 &   1.5657 &    3.0108 \\
\hline \hline
    \end{tabular} }
    \caption{This table reports some statistics for the CDS market spreads and control variables. Transition risk has been computed by the Wasserstein distance.}
    \label{tab:descriptive_stats}
\end{table} 
\noindent
The tables below report the correlation between the change in transition risk at different maturities $\Delta \mbox{TR}_t^m$ (computed via the Wasserstein distance \ref{tab:correlation_CP_weiss} and by the median CDS spread \ref{tab:correlation_CP_med}), and the daily first-difference of the Carbon Price. 
\begin{table}[H]
\centering
 \scalebox{0.8}{
\begin{tabular}{l|llllllllll} \hline
& $\Delta$ TR 6m & $\Delta$ TR 1y & $\Delta$ TR 2y & $\Delta$ TR 3y & $\Delta$ TR 4y & $\Delta$ TR 5y & $\Delta$ TR 7y & $\Delta$ TR 10y & $\Delta$ TR 20y & $\Delta$ TR 6m \\  \hline
$\Delta \mbox{CP}_t$  & -0.0939        & -0.0789        & -0.1067        & -0.1200        & -0.1223        & -0.1170        & -0.1171        & -0.1144         & -0.1079         & -0.1041           \\ \hline 
\end{tabular}}
\caption{Correlation between $\Delta \mbox{CP}_t$ and $\Delta \mbox{TR}_t$, computed by the Wasserstein distance.}
\label{tab:correlation_CP_weiss}
\end{table}

\begin{table}[H]
\centering
 \scalebox{0.8}{
\begin{tabular}{l|llllllllll} \hline
& $\Delta$ TR 6m & $\Delta$ TR 1y & $\Delta$ TR 2y & $\Delta$ TR 3y & $\Delta$ TR 4y & $\Delta$ TR 5y & $\Delta$ TR 7y & $\Delta$ TR 10y & $\Delta$ TR 20y & $\Delta$ TR 6m \\  \hline
$\Delta \mbox{CP}_t$ &    -0.0875     & -0.0857       & -0.1031        & -0.0770        & -0.0650        & -0.0879       & -0.1038        & -0.1036         & -0.0982        & -0.0985           \\ \hline 
\end{tabular}}
\caption{Correlation between $\Delta \mbox{CP}_t$ and $\Delta \mbox{TR}_t$, computed by median CDS spread.}
\label{tab:correlation_CP_med}
\end{table}

\end{appendices}
\bibliography{ArXiv_LivieriSmaniottoRadi.bib}

\end{document}